%% file: main.tex
\documentclass[11pt]{article}

\usepackage{graphicx}
\usepackage{fullpage}
\usepackage{amsfonts}
\usepackage{indentfirst}
\usepackage{amsmath}
\usepackage{amsthm}
\usepackage{amssymb}
\usepackage[numbers]{natbib}
\usepackage[pagebackref=true]{hyperref}
\usepackage{cleveref}
\usepackage{indentfirst}
\usepackage{color}
\usepackage[boxed,linesnumbered]{algorithm2e}
\usepackage{enumitem}
\usepackage{tikz}
\usetikzlibrary{arrows,arrows.meta, positioning, calc, shapes.multipart, shapes.misc, shapes.geometric, patterns}

\usepackage{times}

\definecolor{corlinks}{RGB}{0,0,150}
\definecolor{cormenu}{RGB}{0,0,150}
\definecolor{corurl}{RGB}{0,0,150}

\hypersetup{
	colorlinks=true,
	urlcolor=corlinks,
	linkcolor=corlinks,
	menucolor=cormenu,
	citecolor=corlinks,
	pdfborder= 0 0 0
}

\definecolor{blue-violet}{rgb}{0.54, 0.17, 0.89}
\definecolor{orange}{rgb}{0.9, 0.3, 0.0}

\newcommand{\poly}{\mathsf{poly}}

\newcommand{\N}{\mathbb{N}}

\renewcommand{\P}{\mathsf{P}}
\newcommand{\NP}{\mathsf{NP}}

\newcommand{\PV}{\mathsf{PV}}
\newcommand{\UP}{\mathsf{UP}}
\newcommand{\DTIME}{\mathsf{DTIME}}
\newcommand{\SOT}{\mathsf{S}^1_2}

\newcommand{\APC}{\mathsf{APC}}

\newcommand{\FP}{\mathsf{FP}}

\newcommand{\bd}{\mathsf{b}}

\newcommand{\Seq}{\text{Seq}}
\newcommand{\Read}{\text{Read}}
\newcommand{\Len}{\text{Len}}
\newcommand{\CircuitVal}{\text{CircuitVal}}
\newcommand{\abs}[1]{\lvert #1 \rvert}
\newcommand{\ceil}[1]{\lceil #1 \rceil}

\newtheorem{theorem}{Theorem}[section]
\newtheorem{lemma}[theorem]{Lemma}

\newtheorem{claim}[theorem]{Claim}
\newtheorem{corollary}[theorem]{Corollary}
\newtheorem{proposition}[theorem]{Proposition}
\theoremstyle{definition}
\newtheorem*{problem}{Open Problem}

\theoremstyle{definition}
\newtheorem{definition}[theorem]{Definition}

\theoremstyle{remark}

\title{Parallelism and Adaptivity in Student-Teacher Witnessing}
\author{\large Ondřej Ježil\thanks{\texttt{Email: ondrej.jezil@email.cz}} \\ \small Faculty of Mathematics and Physics\\\small Charles University \and \large Dimitrios Tsintsilidas\thanks{\texttt{Email: dimitrios.tsintsilidas@warwick.ac.uk}}\\  \small Department of Computer Science\\ \small University of Warwick}
\date{February 2026}

\begin{document}

\maketitle
\vspace{-0.6cm}

\begin{abstract}
Student-Teacher games are a model of computation where a computationally restricted student tries to find a string satisfying some refutable property, and every time the student outputs a candidate answer an~all-knowing teacher tries to refute it if possible. These games are a classical computational model for  the witnessing of $\forall\exists\forall$-formulas in bounded arithmetic by the well-known result of Krajíček, Pudlák and Takeuti~\cite{KrajicekPT91}. We introduce subclasses of total search problems in the polynomial hierarchy which are characterized by the number of rounds and the number of candidate answers per round a student from a related level of the polynomial hierarchy would need to solve the given problem.  Our main results are as follows: 

\begin{enumerate}[label=(\alph*)]
    \item\label{bulletPVBB} We find theories of bounded arithmetic whose $\forall\exists \Pi^b_i$-consequences are witnessed by Student-Teacher games with a given amount of adaptivity and parallelism. As a consequence, we show that assuming $\NP\not\subseteq \P/\poly$ the theories 
    \[\text{$\PV_1$, $\PV_1+\mathsf{BB}(\Sigma^b_1)$, $\PV_1+\mathsf{LLIND}(\mathsf{s}\Sigma^b_1)$, $\PV_1+\mathsf{LLIND}(\Sigma^b_1)$ and $\mathsf{S}^1_2$,}\]
    are all distinct. This extends the results of Zambella~\cite{Zambella96}, Cook and Thapen~\cite{DBLP:journals/tocl/CookT06} and Garlík~\cite{garlik2015model} who separated each of the pairs in this sequence under stronger or seemingly incomparable assumptions.
    \item Generalizing the previous to all levels of polynomial hierarchy and assuming $\Sigma^p_{i+1}\not\subseteq \Delta^p_{i+1}/\poly$, we obtain an analogous result for $\mathsf{T}^i_2$. Thus, under plausible complexity theoretic assumptions, we give a solution to an open problem of Buss and Ressayre~\cite{buss1985bounded,clote1993open} regarding the strength of the $\mathsf{BB}(\Sigma^b_{i+1})$ scheme, and another problem of Pollett~\cite{pollett1997arithmetic} regarding the strength of the schemes $\mathsf{LLIND}(\Sigma^b_{i+1})$ and $\mathsf{LLIND}(\mathsf{s}\Sigma^b_{i+1})$.
    \item We revisit two major unconditional unprovability results for $\PV_1$. Namely the unprovability of circuit upper bounds of Krajíček and Oliveira~\cite{KO17}, which we show holds for the theory $\PV_1+\mathsf{BB}(\Sigma^b_1)$ and the unprovability of strong co-nondeterministic circuit lower-bounds of Pich and Santhanam~\cite{pichS21}, which we show holds for the theory $\PV_1+\mathsf{LLIND}(\mathsf{s}\Sigma^b_1)$.  It follows that both of these unprovability results hold simultaneously for a specific theory which extends $\PV_1$, and this theory is a proper extension of $\PV_1$ assuming $\NP\not\subseteq \P/\poly$.
    \item As a technical tool, we build upon the result of Krajíček, Pudlák and Sgall~\cite{KrajicekPS90} and assuming $\Sigma^p_{i+1}\not\subseteq \Delta^p_{i+1}/\poly$ we separate different classes of $\mathsf{TF}\Sigma^p_{i+2}$ which can be solved in a Student-Teacher game with $r$-many rounds (adaptivity) and $q$-many candidate answers per round (parallelism). Notably, the addition of a single round cannot be substituted by polynomial blow-up in the number of parallel answers. We also provide a general witnessing theorem for $\forall\exists\forall$-consequences of first-order theories, which we use as a base result for our other witnessing theorems.
\end{enumerate}
\end{abstract}

\newpage
\tableofcontents

\newpage

\input{introduction.tex}

\paragraph{Acknowledgements.} We would like to thank Igor Oliveira for hosting Ondřej Ježil at the University of Warwick and initiating the discussion about witnessing theorems, Student-Teacher Games and unprovability results. We are also grateful to Jan Kraj{\'\i}{\v{c}}ek for pointing out the provability of the bounded collection axiom with double length induction on non-strict formulas, Sam Buss for pointing us to the open problem from his thesis, and Neil Thapen for discussion about the method of witnessing described in his thesis.

Ondřej Ježil was supported by Charles University Research Center program No. UNCE/24/SCI/022, the project SVV-2025-260837, and by the GA UK project No. 246223. Dimitrios Tsintsilidas was supported by the UKRI Frontier Research Guarantee Grant EP/Y007999/1, the Centre for Discrete Mathematics and its Applications (DIMAP) at the University of Warwick, and Chancellor’s Scholarship.

\input{preliminaries.tex}

\input{student-teacher.tex}
    
\input{witnessing.tex}

\input{separations.tex}

\input{unprovability.tex}

\bibliographystyle{alpha}
\bibliography{ref}

\appendix
\input{appendix.tex}

\end{document}

%% file: introduction.tex
\section{Introduction}

\subsection{Context and Motivation}

Bounded Arithmetic refers to weak theories of arithmetic that were introduced to study the relation of computational complexity with mathematical logic. The theories in Buss' hierarchy are the most used theories for this goal, since they closely relate and formalise reasoning with functions in the corresponding class of the polynomial hierarchy \cite{buss1985bounded}. In fact, the theories $\mathsf{S}^i_2,\mathsf{T}^i_2$ of the $i$th level of the hierarchy are axiomatisable by different axiom schemes, length induction ($\mathsf{LIND}$) and usual induction ($\mathsf{IND}$), respectively, on predicates from the class $\Sigma^p_i$ (in logic as bounded formulas $\Sigma^b_i$). At the same time, by Buss's theorem all the functions with $\Sigma^b_i$ graphs which are provably total in the theory $\mathsf{S}^i_2$ are exactly the functions in the class $\mathsf{FP}^{\Sigma^p_{i-1}}$. By their definition, it is easy to prove that $\mathsf{S}^i_2 \subseteq\mathsf{T}^i_2$, and as expected by the polynomial hierarchy, we also have $\mathsf{T}^{i}_2 \subseteq\mathsf{S}^{i+1}_2$, thus getting an increasing sequence of theories.

Another widely used Bounded Arithmetic theory is $\PV_1$, which is placed at the bottom of Buss's hierarchy. This theory was introduced by Stephen Cook \cite{Coo75} as the equational theory $\PV$, based on Cobham's axioms for defining symbols for all $\FP$ functions \cite{Cob64}, and later extended to first-order as a fragment of Peano Arithmetic \cite{KrajicekPT91} with open formula ($\P$-time predicate) induction. $\PV_1$ is proven to be equivalent to $\mathsf{T}^0_2$, if we consider an extended language for the latter \cite{jerabek:sharply-bounded}, but they are also equivalent theories if we get $\mathsf{T}^0_2$ on the language of $\PV$.

The question if Buss's hierarchy is proper was addressed very early in the literature. The separations between these theories rely on a fundamental tool of Bounded Arithmetic, the witnessing theorems. These theorems show that from a proof of a suitable existential statement in a given theory, one can extract a computational procedure that ``witnesses'' the existential quantifier. In \cite{KrajicekPT91}, a new witnessing theorem is developed, which was later called KPT theorem, after the authors, in order to show that the inclusion $\mathsf{T}^{i}_2 \subseteq\mathsf{S}^{i+1}_2$ (or $\PV_1 \subseteq\mathsf{S}^{1}_2$) is proper, under the assumption that $\Sigma^p_{i+1}\not\subseteq \Delta^p_{i+1}/\poly$. In general, if the polynomial hierarchy does not collapse, then Buss's hierarchy does not collapse, as well. The case of $\mathsf{S}^i_2$ versus $\mathsf{T}^i_2$ was addressed in \cite{krajivcek1993fragments}, where another witnessing theorem is used, and the separation is achieved under the assumption that $\P^{\Sigma^p_i}[O(\log n)]$ (logarithmic number of $\Sigma^p_i$-oracle queries) is not equal to $\Delta^p_{i+1}$ (which is $\P^{\Sigma^p_i}$ with unlimited number of queries).

The KPT theorem is particularly interesting, because the computational procedure extracted by a proof is represented by an interactive game, called Student-Teacher Game. In this game, the Student tries to witness the existential quantifier of a $\exists\forall$-formula and sends their attempts to the Teacher, who either accepts if the Student guessed the correct answer or otherwise, provides a counterexample, which corresponds to the universal quantifier. The Student then can use the counterexamples adaptively to guess a better answer. The KPT theorem asserts that if the statement is provable in the theory $\PV_1$ or some theory $\mathsf{T}^i_2$, there is a Student represented by uniform functions in the corresponding level of the polynomial hierarchy that succeeds in a constant number of rounds. 

This framework of Student-Teacher Games has been proved useful for other applications, too, such as unprovability results. For example, two important unprovability results are that for any $k \in \N$, $\PV_1 \nvdash \P \subseteq \mathsf{SIZE}[n^k]$ due to Kraj\'{\i}\v{c}ek-Oliveira \cite{KO17}, and also that $\PV_1\nvdash \mathsf{NSubExp} \not\subseteq \mathsf{Avg}\text{-} \mathsf{coNSIZE}[2^{n^\delta}]$\footnote{This unprovability statement can also be extended to a fragment of the theory $\APC_1$, but in this paper we only consider theories in Buss's hierarchy.} due to Pich–Santhanam \cite{pichS21}. These results belong to a wider family of results that are concerned with the meta-mathematics of complexity theory \cite{oli2025}, which has as its main goal to prove known results of algorithms and complexity in the weakest theories possible (e.g. $\PV_1$ proves the Cook-Levin theorem and PCP theorem \cite{pichpcp}, $\mathsf{T}^2_2$ proves the circuit size hierarchy \cite{ckkot}, etc.), or to establish the unprovability of open problem statements in these theories. It is worth noting that, by~\cite{KrajicekPT91}, the separation of the theories in Buss's hierarchy is itself equivalent to the unprovability of certain open problem statements about the polynomial-time hierarchy.

Another axiom considered in \cite{buss1985bounded} is the bounded replacement axiom, which is also called collection or choice axiom, since it can be considered as a feasible analogue to the axiom of choice. It is useful because it provides a theory with a quantifier exchange property, where you can move all the bounded quantifiers of the formula outside the sharply bounded quantifiers (quantifiers with smaller scope). This is in analogy with the case of unbounded fragments of Peano Arithmetic, where the replacement axiom can change the order of unbounded and bounded quantifiers. In fact in this case, the theory axiomatised by the collection axiom on $\Sigma_{i+1}$ formulas (with $i+1$ unbounded quantifiers) is strictly between the theories axiomatised by induction on $\Sigma_{i}$ and $\Sigma_{i+1}$ formulas \cite{PARIS1978199}. We write that as $\mathsf{I}\Sigma_{i}\subsetneq \mathsf{B}\Sigma_{i+1}\subsetneq \mathsf{I}\Sigma_{i+1}$. We expect the same to be true in the bounded case, too.

As a matter of fact, Buss proved the analogous inclusions $\mathsf{S}^{i}_2\subseteq \SOT +\mathsf{BB}(\Sigma^b_{i+1})\subseteq \mathsf{S}^{i+1}_2$, where $\SOT$ is used as a base theory and $\mathsf{BB}(\Sigma^b_{i+1})$ denotes the bounded replacement axiom. However, it was left open whether these inclusions are proper. It was later proposed as a major open problem in Bounded Arithmetic by Buss and Ressayre (Problem 12 in Section 1.2 of the survey on open problems in Bounded Arithemetic, Proof Complexity and Fragments of Peano Arithmetic \cite{clote1993open}).

The answer for the second part was given by Zambella \cite{Zambella96}, who shows by a model theoretic argument that $\mathsf{S}^{i+1}_2\subseteq\mathsf{T}^i_2+\mathsf{BB}(\Sigma^b_{i+1})$ is equivalent with $ \mathsf{S}^{i+1}_2\subseteq\mathsf{T}^i_2$. Thus, he separated the bounded replacement axiom from the length induction axiom for the same class of formulas, under the same assumptions \cite{KrajicekPT91} uses to separate $\mathsf{T}^i_2$ and $\mathsf{S}^{i+1}_2$. Nonetheless, the first part, whether $\mathsf{S}^i_2\vdash \mathsf{BB}(\Sigma^b_{i+1})$ was not considered, or the even stronger, whether $\mathsf{T}^i_2\vdash \mathsf{BB}(\Sigma^b_{i+1})$. This question has been open since then.

The first answer towards that direction was given by Cook and Thapen \cite{DBLP:journals/tocl/CookT06}. Here, they prove that under the hardness of factoring against probabilistic polynomial time, $\PV_1\not\vdash \mathsf{BB}(\Sigma^b_1)$. Also, using similar model theoretic ideas with \cite{Zambella96}, they show a new witnessing theorem for the theory $\PV_1+\mathsf{BB}(\Sigma^b_0)$, where we get a version of the Student-Teacher Game with constant rounds and polynomially many parallel queries per round. Provided this witnessing theorem, Je\v{z}il proved the separation of $\PV_1+\mathsf{BB}(\Sigma^b_0)$ from $\SOT$ under the new cryptographic assumption that non-uniform polynomial-time algorithms cannot factor a constant fraction of all semiprimes \cite{jezil}.

However, the separation from \cite{DBLP:journals/tocl/CookT06} is only known for the first level of the hierarchy and cannot generalise to all levels, due to the nature of the factoring problem. Hence, more general hardness assumptions would be preferable to completely solve the problem. 

Another problem, concerning theories between $\PV_1$ and $\SOT$ comes from the variant of length induction axiom, where length is replaced by double length ($\mathsf{LLIND}$). It is not known whether $\mathsf{LLIND}$ accepted for the bounded existential formulas in the language of $\PV$ proves the suitable quantifier exchange property. The generalised version about the separation of theories axiomatised by $\mathsf{LLIND}(\Sigma^b_{i+1})$ and, on the other hand, $\mathsf{LLIND}(\mathsf{strict}\Sigma^b_{i+1})$ has been a long-standing open problem mentioned in the works of Pollett \cite{pollett1997arithmetic,POLLETT1999189,pollett2018finite}.

Only a partial answer to the problem was known by Garlík, who shows by a model theoretic argument that on the first level of the hierarchy the corresponding theories are actually separated, assuming the existence of one-way permutations secure against polynomial-size circuits \cite{garlik2016construction}. Nonetheless, it was still open to prove this, under worst-case assumptions, or to generalise this to all levels of Buss's hierarchy.

\paragraph{Our Contributions.}
In this paper, motivated by the aforementioned open problems on separations of theories and the extensive use of the Student-Teacher witnessing in Bounded Arithmetic, we show the following:
\begin{enumerate}
    \item\textbf{Solutions to open problems about separations of theories:} We give conditional solutions for all levels of Buss's hierarchy to the two aforementioned open problems about fundamental separations of theories in Bounded Arithmetic \cite{clote1993open,pollett1997arithmetic}. The condition we use is the same as in \cite{KrajicekPT91} about the separation of theories $\mathsf{T}^i_2$ and $\mathsf{S}^i_2$. Our results extend the results of \cite{DBLP:journals/tocl/CookT06,garlik2015model}, who solved the two problems only on the first level of the hierarchy using stronger or seemingly incomparable assumptions.

    \item\textbf{Generalised separation of theories:} We provide a bigger family of different theories between $\PV_1$ and $\SOT$, which form hierarchies over variations of the length induction axioms and over the bounded replacement axioms. All these theories are separated under the assumption $\NP\not\subseteq \P/\poly$ (see \Cref{fig:theory-poset}). Again, these separations are generalised to all levels of Buss's hierarchy, where the assumption becomes $\Sigma^p_{i+1}\not \subseteq \Delta^p_{i+1}/\poly$.

    \item\textbf{Lifting unprovability results:} As an attempt towards Open Problems 5.1 and 5.8(d) in \cite{oli2025}, we lift the unprovability results $\PV_1 \nvdash \P \subseteq \mathsf{SIZE}[n^k]$ of Krajíček-Oliveira \cite{KO17}, and $\PV_1 \nvdash \mathsf{NSubExp} \not\subseteq \mathsf{Avg-} \mathsf{coNSIZE}[2^{n^\delta}]$ of Pich-Santhanam \cite{pichS21} to the theories $\PV_1+\mathsf{BB}(\Sigma^b_1)$ and $\PV_1+\mathsf{LLIND}(\mathsf{s}\Sigma^b_1)$, respectively. These theories, and also a theory in the intersection where both results hold, are strictly stronger than $\PV_1$, under $\NP\not\subseteq \P/\poly$, from the previous results.
    
    \item\textbf{Separations of Student-Teacher Games classes:} Our main technical tool is the study of Student-Teacher Games. For that, we classify total search problems in the polynomial hierarchy into new complexity classes we introduce, which are characterised by the number of rounds and the number of candidate solutions per round that a Student needs to solve the problem in the Student-Teacher Game manner. Then, building on the result of \cite{KrajicekPS90}, under the assumption that $\Sigma^p_i\not\subseteq \Delta^p_i/\poly$, we separate these classes for both the direction of rounds and the direction of parallel queries. These separations will be the core for the separation of different theories, while the fine-grained version of these classes will be the one that will improve the unprovability results.

    \item\textbf{Unifying witnessing theorems:} In order to connect these fine-grained Student-Teacher classes with the corresponding theories, we provide a model-theoretic proof that unifies all the witnessing theorems of this kind. To use this general result, we also provide the description of these classes in an appropriate theory.
\end{enumerate}

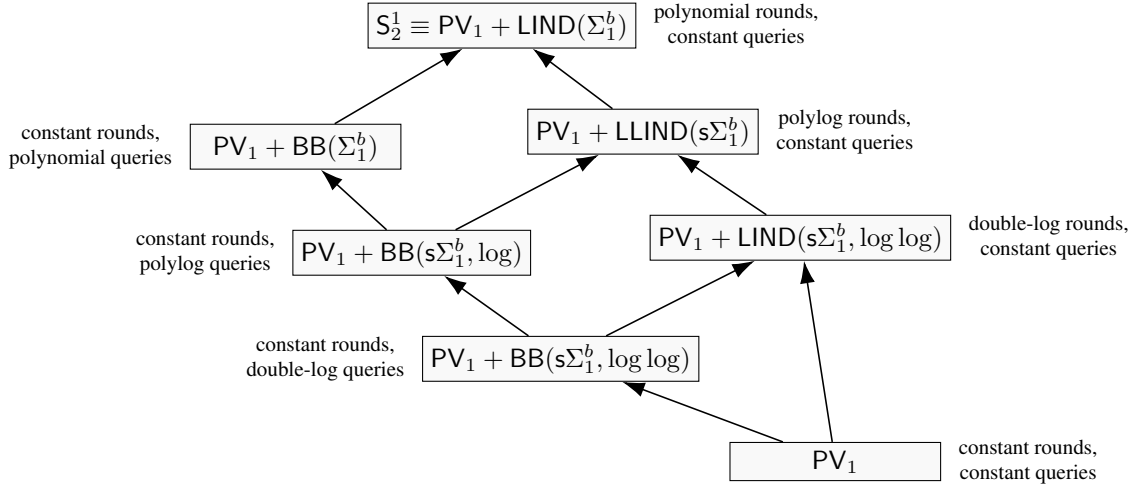
\begin{figure}[ht]
\centering
\begin{tikzpicture}[
    node distance=12mm and 20mm,
    theory/.style = {draw, rounded corners=0pt, top color=gray!5, bottom color=gray!5,
                     font=\small, minimum width=28mm, align=center, inner sep=3pt},
    arrow/.style = {-{Latex[length=3mm,width=2mm]}, line width=0.6pt},
    comment/.style = {font=\scriptsize, align=center}
  ]

  \node[theory] (LIND) {$\mathsf{S}^{1}_2 \equiv \PV_1 + \mathsf{LIND}(\Sigma^b_{1})$};
  \node[comment, right=1mm of LIND] {polynomial rounds, \\ constant queries};

  \node[theory, below left=10mm and 0.5mm of LIND, xshift = 5mm] (BBpoly) {$\PV_1 + \mathsf{BB}(\Sigma^b_{1})$};
  \node[comment, left=1mm of BBpoly] {constant rounds, \\ polynomial queries};

  \node[theory, below right=8mm and 0.5mm of LIND, xshift = -15mm] (LINDlog) {$\PV_1 + \mathsf{LLIND}(\mathsf{s}\Sigma^b_{1})$};
  \node[comment, right=1mm of LINDlog] {polylog rounds, \\ constant queries};

  \node[theory, below right=8mm and 0.5mm of LINDlog, xshift = -15mm] (LINDloglog) {$\PV_1 + \mathsf{LIND}( \mathsf{s}\Sigma^b_{1},\log\log)$};
  \node[comment, right=1mm of LINDloglog] {double-log rounds, \\ constant queries};

  \node[theory, below right=8mm and 0.5mm of BBpoly, xshift = -15mm] (BBlog) {$\PV_1 + \mathsf{BB}(\mathsf{s}\Sigma^b_{1},\log)$};
  \node[comment, left=1mm of BBlog] {constant rounds, \\ polylog queries};

  \node[theory, below right=8mm and 0.5mm of BBlog, xshift = -15mm] (BBloglog) {$\PV_1 + \mathsf{BB}(\mathsf{s} \Sigma^b_{1},\log\log)$};
  \node[comment, left=1mm of BBloglog] {constant rounds, \\ double-log queries};

  \node[theory, below right=8mm and 4mm of BBloglog] (PV1) {$\PV_1$};
  \node[comment, right=1mm of PV1] {constant rounds, \\ constant queries};

  % Inclusions (arrows)
  \draw[arrow] (PV1) -- (LINDloglog);
  \draw[arrow] (PV1) -- (BBloglog);
  \draw[arrow] (BBloglog) -- (BBlog);
  \draw[arrow] (BBloglog) -- (LINDloglog);   
  \draw[arrow] (BBlog) -- (BBpoly);
  \draw[arrow] (BBlog) -- (LINDlog);
  \draw[arrow] (BBpoly) -- (LIND);
  \draw[arrow] (LINDloglog) -- (LINDlog);
  \draw[arrow] (LINDlog) -- (LIND);

\end{tikzpicture}
\caption{Partial order of theories between $\PV_1$ and $\mathsf{S}^{1}_2$. There are two hierarchies, one with extensions axiomatised by length induction axioms (right), and one with extensions axiomatised by bounded replacement axioms (left). The levels are for the bounds $\poly,\log,\log\log$, etc.
Under the hardness assumption $\NP\not\subseteq \P/\poly$, all these theories are distinct. The arrows symbolise the only possible inclusions. }
\label{fig:theory-poset}
\end{figure}

\subsection{Our results}

In what follows, we give an overview of the main results from the paper.

\subsubsection{Student--Teacher Games} 
Our main technical tool is the computational framework provided by Student-Teacher Games, where we establish lower bounds against different variants of them. We start by considering them as a computational model for search problems in $\mathsf{TFPH}$. We adopt the definition of $\mathsf{TF\Sigma^p_{i+2}}$ for $i\geq 0$, from \cite{KleinbergKMP21}.

\begin{definition}
A relation $R(x, y)$ belongs to $\mathsf{TF\Sigma^p_{i+2}}$ if it is polynomially bounded and total (that is, for every $x$ there exists $y$ such that $(x, y)$ is in the relation), and there exist a predicate $\varphi\in \Delta^p_{i+1}$ and a polynomial $p(n)$ such that
\[
R(x, y) \iff \forall z \in \{0,1\}^{p(|x|)}\, \varphi(x, y, z).
\]
\end{definition}

We now define the notion of a \emph{Student--Teacher Game} and introduce the corresponding complexity classes.

\begin{definition}[Student--Teacher Game]
A $\mathsf{TF\Sigma^p_{i+2}}$ relation $R(x,y)$, specified by a polynomial-time predicate $\varphi\in \Delta^p_{i+1}$ and a polynomial $p(n)$, can be solved by a Student-Teacher Game with $r(n)$ rounds and $q(n)$ parallel queries per round if, for $n=|x|$, there exists a uniform function $f\in\FP^{\Sigma^p_i}$, representing a Student, that takes as input the number $x$ of size $n$, the index of the round $i\in [r(n)]$, and an additional input depending on the round, which has size $(i-1)\cdot p(n)\cdot q(n)$, and produces outputs $y^i_1, y^i_2, \dots, y^i_{q(n)}$, such that:
\begin{multline}
\forall z^1 = (z^1_1,\dots,z^1_q), \dots, z^{r} = (z^{r}_1,\dots,z^{r}_q) \;
\big(\exists j \in [q]\; \varphi(x,f(x,1)_j,z^1_j)\big)
\;\lor\\
\big(\exists j \in [q]\, \varphi(x,f(x,2,z^1)_j,z^2_j)\big)
\;\lor \dots \lor\;
\big(\exists j \in [q]\; \varphi(x,f(x,r,z^1,\dots,z^{r-1})_j,z^{r}_j)\big),
\nonumber
\end{multline}
where we abbreviate $r := r(n)$ and $q := q(n)$, and we consider these as functions computable in polynomial time for the sake of well-definedness.

The class $\mathcal{ST}^{\Sigma^p_i}[r(n),q(n)]$ consists of all $\mathsf{TF\Sigma^p_{i+2}}$ relations $R(x,y)$ that can be solved by a Student-Teacher Game with $r(n)$ rounds and $q(n)$ parallel queries per round.
\end{definition}

The interpretation of this computational model is as an interaction between a Student trying to witness the existential quantifier of the $\mathsf{TF\Sigma^p_{i+2}}$ problem and a Teacher that provides counterexamples for the next universal quantifier. Thus, the Student represented by the $\FP^{\Sigma^p_i}$ function $f$, at every round receives the input, the number of the round and the corresponding counterexamples provided by the Teacher for their previous attempts, and computes $q(n)$ new candidate witnesses. If the round is more than $r(n)$, then the Student must reject. Also, the Teacher's counterexamples can be checked by the Student using the $\Sigma^b_i$-oracle.

For Student-Teacher Games with one query per round, we have the following foundational theorem from \cite{KrajicekPS90}, where the framework was first introduced.

\begin{theorem}[\cite{KrajicekPS90}]\label{thm:sgall}
If $\NP \not\subseteq \P/\mathsf{poly}$, then for any sublinear, unbounded, increasing, polynomial-time function $r(m)$,
\[
\mathcal{ST}[r(m),1] \subsetneq \mathcal{ST}[r(m)+1,1].
\]
\end{theorem}

The theorem can be easily generalised to $\mathcal{ST}^{\Sigma^p_i}[r(m)-1,1] \subsetneq \mathcal{ST}^{\Sigma^p_i}[r(m),1]$ under the assumption $\Sigma^p_{i+1}\not\subseteq\Delta^p_{i+1}/\poly$. However, from this theorem, it is not clear where this power of the extra round is coming from; the adaptivity through the interaction between the Student and the Teacher or the number of counterexamples that the Student is able to use.

Motivated by this question, we study the Student-Teacher Games with multiple (up to polynomially many) parallel queries per round and we extend the aforementioned result by establishing separations in both the \emph{round} and \emph{query} dimensions. That is both adaptivity through interaction and number of parallel queries (along the respective counterexamples) can give more power to the Student.

Our first result generalizes \Cref{thm:sgall} by showing that the separation between $r(m)$ and $r(m)+1$ rounds persists even when the Student may issue an arbitrary polynomial number of parallel queries. 
This demonstrates that \emph{adaptivity across rounds} is a crucial source of power in Student-Teacher interactions.

\begin{theorem}\label{thm:st1}
If $\Sigma^p_{i+1}\not\subseteq\Delta^p_{i+1}/\poly$, then for any sublinear, unbounded, increasing, polynomial-time function $r(m)$ and any unbounded, increasing, polynomial-time function $1\leq q(m) \leq \mathsf{poly}(m)$,
\[
\mathcal{ST}^{\Sigma^p_i}[r(m)+1,1] \not\subseteq\mathcal{ST}^{\Sigma^p_i}[r(m),q(m)].
\]
\end{theorem}

The idea for the proof is that there exists a hard search problems, which can be solved in $r(m)+1$ rounds, but if the Student has one less round, even with many queries per round, they do not get enough information for the teacher. If the opposite is true, then we use this powerful Student to decide any $\Sigma^p_{i+1}$ problem with a $\Delta^p_{i+1}$-algorithm that uses polynomial-size advice. The proof is quite technical and it is resumed for \Cref{sec:st1}.

Our second result highlights the power gained from \emph{parallelism}; that is, from allowing the Student to make multiple simultaneous queries to the Teacher. 
Essentially, each additional query, which results to a counterexample, combined with one extra adaptive round, allows the Student to solve strictly more problems.

\begin{theorem}\label{thm:st2}
If $\Sigma^p_{i+1}\not\subseteq\Delta^p_{i+1}/\poly$, then for any sublinear, unbounded, increasing, polynomial-time function $2\leq q(m)\leq m$,
\[
\mathcal{ST}^{\Sigma^p_i}[2,q(m)-1] \subsetneq \mathcal{ST}^{\Sigma^b_i}[2,q(m)],
\]
and, more generally, for any sublinear, unbounded, increasing functions $1\leq r_1(m), q_1(m), r_2(m), q_2(m)$,
\[
\mathcal{ST}^{\Sigma^p_i}[1+r_1(m),q_1(m)] \not\subseteq \mathcal{ST}^{\Sigma^p_i}[1+r_2(m),q_2(m)]
\;\text{   whenever   }\; r_1(m)q_1(m) > r_2(m)q_2(m).
\]
\end{theorem}

Note that in this theorem, we only consider Student-Teacher Games with more than one round, since otherwise there is no interaction with the Teacher. The idea of the proof is similar with \Cref{thm:st2}, but we use another kind of search problem. The full proof is in \Cref{sec:st2}.

\subsubsection{Witnessing with Student-Teacher Games}
The original motivation for defining Student--Teacher Games arises from the witnessing theorems of Bounded Arithmetic. 
The KPT Theorem, due to Krajíček, Pudlák, and Takeuti~\cite{KrajicekPT91}, is a Herbrand-type result for universal theories, which in our case applies to the theories $\PV_1$ and universal conservative extensions of $\mathsf{T}^i_2$. 
We first recall the general version.

\begin{theorem}[General KPT Theorem]
Let $\mathsf{T}$ be a universal theory in a $\mathcal{L}$, and let $\varphi$ be a $\Sigma_1$-$\mathcal{L}$-formula. 
Suppose that
\[
\mathsf{T} \,\vdash\, \forall x\, \exists y \, \forall z \;\varphi(x,y,z).
\]
Then there exists a finite sequence $s_1, \dots, s_k$ of $\mathcal{L}$-terms such that
\[
\mathsf{T} \,\vdash\,  \forall x,z_1,\dots,z_k\,
\bigl(\varphi(x,s_1(x),z_1) \lor \varphi(x,s_2(x,z_1),z_2) \lor \dots \lor \varphi(x,s_k(x,z_1,\dots,z_{k-1}),z_k)\bigr).
\]
\end{theorem}

The theory $\PV_1$ defined in \cite{KrajicekPT91} is a universal theory whose terms denote polynomial-time functions. In a similar manner, the theories $\PV_{i+1}$ are universal theories whose terms are the $\FP^{\Sigma^b_i}$-functions, but they are conservative over $\mathsf{T}^i_2$ (under the same language they prove the same theorems), thus we do not distinguish them here.
Combining the KPT theorem with the definition of Student-Teacher Games allows us to characterize the $\Sigma^b_{i+2}$-consequences of $\mathsf{T}^i_2$, or for $\PV_1$ if $i=0$.

\begin{theorem}
Let $i\geq0$ and $T$ be $\PV_1$ for $i=0$ and $\mathsf{T}^i_2$ otherwise. Let $\varphi$ be a $\Sigma^b_i$-formula \footnote{One can also apply the general KPT theorem to a $\exists \Pi^b_i$-formula $\varphi$ and with the quantifier $\forall^b z$ unbounded, obtaining a characterisation on a greater variety of consequences for $\mathsf{T}^i_2$; however, in this work we restrict attention to $\Sigma^b_{i+2}$-consequences only, which correspond to total problems in $\mathsf{TF\Sigma^p_{i+2}}$, which enables us to classify them in the classes $\mathcal{ST}^{\Sigma^b_i}[r(m),q(m)]$. It is noted that we allow the Student to be able to check the validity of the counterexamples in this definition of the Games.}, 
and suppose that 
\[
T \,\vdash\, \forall x\, \exists y \, \forall^b z \;\varphi(x,y,z).
\]
Then there exists a polynomial $p(n)$, given by the bound on $\forall^bz$, such that the relation $R(x,y)$\footnote{The relation $R(x,y)$ is guaranteed to be polynomially bounded by Parikh's Theorem \cite{parikh}, which states that if there is a proof with an unbounded existential quantifier, then a Bounded Arithmetic theory can also prove the same theorem with the quantifier being bounded.} characterized by $\varphi$ and $p(n)$ is total, and there exists a constant $k\in\N$ such that $R(x,y)\in \mathcal{ST}^{\Sigma^p_i}[k,1]$.
\end{theorem}

The theories $\mathsf{S}^i_2$ are not universal, but building on the KPT theorem we can obtain an analogous result.

\begin{theorem}[$\mathsf{S}^i_2$-analogue to KPT \cite{krajivcek1992no,Pudlak92}]
Let $\varphi$ be a $\Sigma^b_i$-formula\footnote{Again, this theorem can be shown for $\varphi\in \Sigma^b_{i+1}$ and unbounded $\forall z$, too.}, and suppose that 
\[
\mathsf{S}^{i+1}_2\,\vdash\, \forall x\, \exists y \, \forall^b z \;\varphi(x,y,z).
\]
Then there exists a polynomial $p(n)$, given by the bound on $\forall^bz$, such that the relation $R(x,y)$ characterized by $\varphi$ and $p(n)$ is total, and there exists a polynomial $r(n)$ such that $R(x,y)\in \mathcal{ST}^{\Sigma^p_i}[r(n),1]$.
\end{theorem}

As we see, the $\Sigma^b_{i+2}$-consequences of the theories $\mathsf{T}^i_2$ correspond to total search problems in the class $\mathcal{ST}^{\Sigma^p_i}[O(1),1]$, whereas the $\Sigma^b_{i+2}$-consequences of $\mathsf{S}^{i+1}_2$ correspond to total search problems in the class $\mathcal{ST}^{\Sigma^p_i}[\mathsf{poly}(n),1]$. Therefore, there is a fundamental difference between these theories.

Following the intermediate classes of Student-Teacher Games defined above, we can also consider theories between $\mathsf{T}^i_2$ and $\mathsf{S}^{i+1}_2$ whose $\Sigma^b_{i+2}$-consequences are captured by the corresponding classes of Student-Teacher Games, thus having more fine-grained separations of theories.

It is known that $\mathsf{T}^{i}_2\subseteq \mathsf{S}^{i+1}_2$, but if we weaken the axioms of $\mathsf{S}^{i+1}_2$ to get weaker theories corresponding to weaker classes of Student-Teacher Games, it may be the case that these theories do not include $\mathsf{T}^{i}_2$ any more. That is why we work with $\mathsf{T}^{i}_2$ as the base theory, and define all other theories as its extensions. To begin with, $\mathsf{S}^{i+1}_2$ can be regarded as an extension of $\mathsf{T}^{i}_2$ by the Length Induction Axiom Scheme for $\Sigma^b_{i+1}$ formulas, denoted $\mathsf{LIND}(\Sigma^b_{i+1})$. 
The formal definition of the axiom for some $\varphi\in \Sigma^b_{i+1}$ is 
\[
\mathsf{LIND}(\varphi) := 
\forall x \Big[(\varphi(0) \land \forall n < |x| \, (\varphi(n)\rightarrow \varphi(n+1))) \rightarrow \varphi(|x|)\Big].
\]

This scheme is analogous to standard induction, but the induction is carried out on the length of the variable rather than on the variable itself. 
Since $\varphi$ is $\Sigma^b_{i+1}$, the axiom $\mathsf{LIND}(\varphi)$ is a $\forall\Sigma^b_{i+2}$ sentence, which can be trivially witnessed in $\mathcal{ST}^{\Sigma^b_i}[\poly(|x|),1]$. We can easily see that if the Student starting from $0$ tries to find a contradiction in $\varphi(n)\to \varphi(n+1)$ step by step by acquiring the witnesses for the existential quantifier for $\varphi$ from the Teacher.

We can also define a variant where induction proceeds up to the double length of the number:
\[
\mathsf{LLIND}(\varphi) := 
\forall x \Big[(\varphi(0) \land \forall n < ||x|| \, (\varphi(n)\rightarrow \varphi(n+1))) \rightarrow \varphi(||x||)\Big].
\]
There are two cases to consider here. If $\varphi$ is a \emph{strict} $\Sigma^b_{i+1}$ formula (denoted $\mathsf{s}\Sigma^b_{i+1}$), this means that all the sharply bounded quantifiers are after the bounded quantifiers (which will make the last part verifiable in polynomial time). In this case, this $\forall\Sigma^b_{i+2}$ axiom can be witnessed in $\mathcal{ST}^{\Sigma^p_i}[\log(|x|),1]$, in the same way as in length induction. 
However, if $\varphi$ is not strict — i.e., it contains sharply bounded quantifiers preceding the outer bounded existential quantifier — then polynomial number of parallel versions of the game are required to eliminate all the cases indicated by the sharply bounded quantifiers. This implies an upper bound for the witnessing of $\varphi$ in the class $\mathcal{ST}^{\Sigma^p_i}[\mathsf{polylog}(|x|),\mathsf{poly}(|x|)]$. 
We see that these $\mathsf{poly}(|x|)$ parallel queries are unnecessary in the case of $\mathsf{LIND}$, as they can be simulated by the polynomially many rounds available in that setting.

This distinction between strictly $\Sigma^b_{i+1}$ and general $\Sigma^b_{i+1}$ corresponds to the exchange property for sharply bounded quantifiers  guaranteed by Bounded Replacement Axiom Scheme ($\mathsf{BB}(\Sigma^b_{i+1})$). 
The $\mathsf{BB}$ axiom for a formula $\varphi\in \Sigma^b_{i+1}$, is defined as follows:
\[
\mathsf{BB}(\varphi) := 
\forall x,t \Big[
(\forall j \leq |x| \, \exists y \leq t \; \varphi(j,y))
\;\rightarrow\;
(\exists w \, \forall j \leq |x| \, \varphi(j, w_j))
\Big],
\]
where $w$ encodes a sequence of length $|x|$ with numbers less than $t$. 
When a theory includes the axiom $\mathsf{BB}(\Sigma^b_{i+1})$, one can show that for any $\Sigma^b_{i+1}$ formula $\psi$, there exists an equivalent strict $\Sigma^b_{i+1}$ formula $\psi'$, which is provable in the theory.

The axioms $\mathsf{BB}(\Sigma^b_{i+1})$ are $\forall\Sigma^b_{i+2}$ sentences that can be witnessed in $\mathcal{ST}^{\Sigma^p_i}[O(1),\mathsf{poly}(n)]$. 
Moreover, they are provable in $\mathsf{S}^{i+1}_2$ and $\mathsf{T}^i_2+\mathsf{LLIND}(\Sigma^b_{i+1})$ \cite{ALLEN19911}, which aligns with the intuition we have from the Student-Teacher Games.

It is proven in \cite{DBLP:journals/tocl/CookT06} that the $\Sigma^b_2$-consequences of $\PV_1+\mathsf{BB}(\Sigma^b_1)$ correspond to total search problems in $\mathcal{ST}[O(1),\mathsf{poly}(n)]$. The proof can be generalised to cover the $\Sigma^b_{i+2}$-consequences of $\mathsf{T}^i_2+\mathsf{BB}(\Sigma^b_{i+1})$, which give us total search problems in $\mathcal{ST}^{\Sigma^p_i}[O(1),\mathsf{poly}(n)]$.

Cook and Thapen also introduced fine-grained variants of $\mathsf{BB}(\Sigma^b_0)$ parametrised by different growth rates of lengths. For that we can use any sublinear, increasing function $\bd(x)$ provably total in $\PV_1$, and we get the axiom:
\[
\mathsf{BB}(\varphi,\bd):=  \forall x,t \Big[(\forall j \leq \bd(|x|) \, \exists y \leq t \; \varphi(j,y))\;\rightarrow\;(\exists w \, \forall j \leq \bd(|x| )\, \varphi(j, w_j)) \Big].
\]
They show that witnessing in $\PV_1+\mathsf{BB}(\Sigma^b_0,\bd)$ corresponds to total search problems in $\mathcal{ST}[O(1),\mathsf{poly}(\bd(n))]$, and they separate all these different theories assuming the hardness of factoring against probabilistic algorithms.

In this work, we also consider a fine-grained version of the length induction axioms, denoted $\mathsf{LIND}(\Sigma^b_{i+1}, \bd)$, and defined as:
\[
\mathsf{LIND}(\varphi,\bd):= \forall x \Big[\big(\varphi(0) \land \forall n < \bd( |x|)\, (\varphi(n)\rightarrow \varphi(n+1))\big) \rightarrow \varphi(\bd(|x|))\Big].
\]
In this way, $\mathsf{LLIND}(\Sigma^b_{i+1})$ is the same as $\mathsf{LIND}(\Sigma^b_{i+1},\log)$ and $\mathsf{LIND}(\Sigma^b_{i+1})$ is the same as $\mathsf{LIND}(\Sigma^b_{i+1},\poly)$.

We aim to separate all the theories which extend $\mathsf{T}^i_2$ with these different axioms, by showing that their $\Sigma^b_{i+2}$-consequences correspond to the natural Student-Teacher Game class. We already know these for $\mathsf{T}^i_2, \,\mathsf{S}^{i+1}_2$ and $\mathsf{T}^i_2+\mathsf{BB}(\Sigma^b_{i+1})$. However, in this work, using ideas from \cite{AVIGAD2002219,thapen2002weak,krajivcek1992no}, we provide a generalised model-theoretic argument that unifies all these witnessing theorems, and enables us to show the witnessing for all these extensions. This framework has the potential to be used in various situations to facilitate the proof of other witnessing theorems, too.

Informally, we take a theory $T$ which is an extension of $\mathsf{T}^i_2$ by an axiomatic scheme with formula complexity in $\Sigma^b_{i+2}$. We also suppose that the axiomatic scheme can be witnessed by an $\FP^{\Sigma^b_i}$-Student in a specific class of Student-Teacher Games that must be closed under definition by cases and finite composition. Then the $\Sigma^b_{i+2}$-consequences of $T$ can also be witnessed within the same class of Student-Teacher Games. The detailed exposition is in \Cref{sec:witnessing}.

\subsubsection{Separations of Theories}

We may thus consider all extensions of $\mathsf{T}^i_2$ by the different axiomatic schemes $\mathsf{LIND}(\mathsf{s}\Sigma^b_{i+1},\bd)$ and $\mathsf{BB}(\mathsf{s}\Sigma^b_{i+1},\bd)$. Each one admits a witnessing theorem within the naturally corresponding $\mathcal{ST}^{\Sigma^p_i}$ class. 
Applying the ideas underlying \Cref{thm:st1,thm:st2}, we obtain the following separations.

\begin{theorem}\label{thm:bbsep}
    If $\Sigma^p_{i+1}\not\subseteq \Delta^p_{i+1}/\poly$, then for all increasing $\PV$ functions $\bd_1(x), \bd_2(x) \leq x$, such that for any $k\in\N$ and for large enough $x$, $\bd_1(x)\geq \bd_2(x)^k$, $\mathsf{T}^i_2+\mathsf{BB}(\mathsf{s}\Sigma^b_{i+1},\bd_2)\not\vdash \mathsf{BB}(\mathsf{s}\Sigma^b_{i+1},\bd_1)$, or equivalently, $\mathsf{T}^i_2+\mathsf{BB}(\mathsf{s}\Sigma^b_{i+1},\bd_1)\not\subseteq \mathsf{T}^i_2+\mathsf{BB}(\mathsf{s}\Sigma^b_{i+1},\bd_2)$.
\end{theorem}

The above theorem is also proved in \cite{DBLP:journals/tocl/CookT06}, but only for the first level, and under a cryptographic hardness assumption instead of a worst-case one.

\begin{theorem}\label{thm:lindsep}
    If $\Sigma^p_{i+1}\not\subseteq \Delta^p_{i+1}/\poly$, then for all increasing $\PV$ functions $\bd_1(x), \bd_2(x) \leq x$, such that for any $k\in\N$ and for large enough $x$, $\bd_1(x)\geq \bd_2(x)^k$, $\mathsf{T}^i_2+\mathsf{LIND}(\mathsf{s}\Sigma^b_{i+1},\bd_2)\not\vdash \mathsf{LIND}(\mathsf{s}\Sigma^b_{i+1},\bd_1)$, or equivalently, $\mathsf{T}^i_2+\mathsf{LIND}(\mathsf{s}\Sigma^b_{i+1},\bd_1)\not\subseteq \mathsf{T}^i_2+\mathsf{LIND}(\mathsf{s}\Sigma^b_{i+1},\bd_2)$.
\end{theorem}

Therefore, we have two hierarchies for the extensions of $\mathsf{LIND}$ and $\mathsf{BB}$, and the levels can be indexed by $\poly,\,\log,\,\log\log,\dots$ etc. We also know that 
\[
\mathsf{T}^i_2+\mathsf{LIND}(\mathsf{s}\Sigma^b_{i+1},\bd)\,\vdash\, \mathsf{BB}(\mathsf{s}\Sigma^b_{i+1},\bd).
\]
by parametrising the proof of $\mathsf{S}^{i+1}_2\vdash\mathsf{BB}(\Sigma^b_{i+1})$ (\Cref{theoremllindprovesbb}), which means that every level of the $\mathsf{LIND}$ hierarchy proves the corresponding theory in the $\mathsf{BB}$ hierarchy. However, this is the best we can have by the two final theorems:

\begin{theorem}\label{thm:bblindsep}
    If $\Sigma^p_{i+1}\not\subseteq \Delta^p_{i+1}/\poly$, then for all increasing $\PV$ functions $\bd(x) \leq x$, $\mathsf{T}^i_2+\mathsf{BB}(\mathsf{s}\Sigma^b_{i+1})\not\vdash \mathsf{LIND}(\mathsf{s}\Sigma^b_{i+1},\bd)$, or equivalently and $\mathsf{T}^i_2+\mathsf{LIND}(\mathsf{s}\Sigma^b_{i+1},\bd)\not\subseteq\mathsf{T}^i_2+\mathsf{BB}(\mathsf{s}\Sigma^b_{i+1})$.
\end{theorem}

\begin{theorem}\label{thm:lindbbsep}
    If $\Sigma^p_{i+1}\not\subseteq \Delta^p_{i+1}/\poly$, then for all increasing $\PV$ functions $\bd_1(x), \bd_2(x) \leq x$, such that for any $k\in\N$ and for large enough $x$, $\bd_1(x)\geq \bd_2(x)^k$, $\mathsf{T}^i_2+\mathsf{LIND}(\mathsf{s}\Sigma^b_{i+1},\bd_2)\not\vdash \mathsf{BB}(\mathsf{s}\Sigma^b_{i+1},\bd_1)$, or equivalently, $\mathsf{T}^i_2+\mathsf{BB}(\mathsf{s}\Sigma^b_{i+1},\bd_1)\not\subseteq \mathsf{T}^i_2+\mathsf{LIND}(\mathsf{s}\Sigma^b_{i+1},\bd_2)$ 
\end{theorem}

The first theorem shows that no theory in the $\mathsf{BB}$ hierarchy includes any theory from the $\mathsf{LIND}$ hierarchy, and the second shows that the inclusions per level are the only ones possible. Therefore, an interesting consequence of that is that the theory $\mathsf{T}^i_2+\mathsf{LIND}(\mathsf{s}\Sigma^b_{i+1},\log)$ is independent from $\mathsf{T}^i_2+\mathsf{BB}(\Sigma^b_{i+1},\poly)$.

The proofs for all the above theorems are applications of \Cref{thm:st1,thm:st2}, where we show that the sentence used for the separation is provable in the corresponding class. We can see all these relations in \Cref{fig:theory-poset}, if we transfer them to the level $i$ of Buss's hierarchy.

\subsubsection{Conditional Solutions to Open Problems in Bounded Arithmetic}

The separations described in the previous section give some natural solutions to two open problems in Bounded Arithmetic regarding the power of the bounded collection axiom and the difference of double length induction for strict and non-strict formulas.\footnote{Amusingly in both problems the concerned theories are both denoted by the letter $\mathsf{R}$. $\mathsf{R}_i= \SOT + \mathsf{BB}(\Sigma^b_i)$ which was defined already in Buss's thesis \cite{buss1985bounded}, and $\mathsf{R}^i_2= \mathsf{BASIC}+\mathsf{LLIND(\Sigma^b_1)}$, which is defined in the same manner as $\mathsf{S}^i_2$ \cite{pollett1997arithmetic}. We do not use these names to avoid confusion.} The difference that these problems have with the above separations is that the corresponding theories are defined over a different base theory. However, the same results still hold. The assumption we use is, as always in the paper, $\Sigma^p_{i+1}\not\subseteq \Delta^p_{i+1}/\poly$, which is the best known assumption for separations of theories in Buss's Hierarchy.

\vspace{-1em}
\paragraph{Buss and Rassayre's Open Problem \cite{clote1993open}.}As mentioned above, it was proven in \cite{buss1985bounded} that over the base theory $\SOT$, the axiomatic scheme $\mathsf{BB}(\Sigma^b_{i+1})$ can be proven by the scheme $\mathsf{LIND}(\Sigma^b_{i+1})$, while it can prove $\mathsf{LIND}(\Sigma^b_{i})$ (see in the Preliminaries \Cref{theorembbproveslind,theoremllindprovesbb,theoremsotprovesbb1}), which means we have $\mathsf{S}^{i}_2\subseteq \SOT +\mathsf{BB}(\Sigma^b_{i+1})\subseteq \mathsf{S}^{i+1}_2$.

The problem is whether these theories can be separated.
It is easy to see that under the usual assumption $\Sigma^p_{i+1}\not\subseteq \Delta^p_{i+1}/\poly$, which is the same as the one used to separate $\mathsf{T}^i_2$ from $\mathsf{S}^{i+1}_2$, we can also show the desired separations. 

By \Cref{thm:bbsep}, for $\bd_2=0$, we get that $\mathsf{T}^i_2\not\vdash \mathsf{BB}(\mathsf{s}\Sigma^b_{i+1},\bd)$ for any $\bd$ which is super-constant, which also implies that $\mathsf{S}^i_2\not\vdash \mathsf{BB}(\mathsf{s}\Sigma^b_{i+1},\bd)$. For $\bd(x)=x$, this directly shows that $\mathsf{BB}(\Sigma^b_{i+1})\not\subseteq \mathsf{S}^i_2$, thus achieving the first separation. 

For the second part, we can use \Cref{thm:bblindsep} with $\bd(x)=x$, which gives us that $\mathsf{T}^i_2+\mathsf{BB}(\Sigma^b_{i+1}) \nvdash \mathsf{LIND}(\mathsf{s}\Sigma^b_{i+1})$, which also means, by weakening the base theory, that $\SOT+\mathsf{BB}(\Sigma^b_{i+1}) \nvdash\mathsf{LIND}(\mathsf{s}\Sigma^b_{i+1})$. This separates the second part.

\vspace{-1em}
\paragraph{Double-length induction on strict vs. non-strict $\Sigma^b_1$-formulas \cite{pollett1997arithmetic}.}The theories concerned in this problem have the basic set of axioms which is used in Buss's theories, denoted by $\mathsf{BASIC}$ (see Preliminaries), but they extend it by two versions of double length induction. As we see in \Cref{theoremllindprovesbb}, the length induction for strict formulas can prove the bounded collection axiom, which provides the sharply bounded quantifier exchange property \Cref{prop:bd-exchange}. This means that length induction on strict formulas is equivalent with length induction with formulas from the corresponding non-strict class. However, this was not known for double length induction; the problem is if double length induction on strict formulas can prove the version for non-strict ones.

We show that this is not true ($\mathsf{BASIC}+\mathsf{LLIND}(\mathsf{s}\Sigma^b_i) \nvdash \mathsf{LLIND}(\Sigma^b_i)$). This is the case, since we already know that $\mathsf{BASIC}+\mathsf{LLIND}(\Sigma^b_i) \nvdash \mathsf{BB}(\Sigma^b_i)$ \cite{ALLEN19911}, and it cannot happen that $\mathsf{BASIC}+\mathsf{LLIND}(\mathsf{s}\Sigma^b_i) \nvdash \mathsf{BB}(\Sigma^b_i)$, under the assumption that $\Sigma^p_{i+1}\not\subseteq \Delta^p_{i+1}/\poly$.  If the latter was true, then we would also have $\PV_1+\mathsf{LLIND}(\mathsf{s}\Sigma^b_i) \vdash \mathsf{BB}(\Sigma^b_i)$, which is contradictory with \Cref{thm:lindbbsep}.

\subsubsection{Unprovability Results} 

In the final section, we revisit two major unprovability results in $\PV_1$ — the unprovability of circuit upper bounds by Krajíček–Oliveira \cite{KO17} and the unprovability of average-case circuit lower bounds by Pich–Santhanam \cite{pichS21} — and investigate whether their conclusions can be strengthened to provably stronger (under the complexity assumptions we use above) theories. It is an open problem if these unprovability results hold also for $\SOT$ (Open Problems 5.1 and 5.8(d) in \cite{oli2025})

Both of these results use the KPT theorem and the Student-Teacher Game derived from it, to find a contradiction. Hence, the question we are trying to answer is what is the strongest Student-Teacher Game class sufficient to obtain contradiction. For the former case, concerning the unprovability of circuit upper bounds, we can see from the proof that having more than constant rounds breaks the argument of composing polynomial-time functions. For the latter case, it is not clear how to adapt the proof for the number of rounds which is polynomial in $2^n$. We were able to extend the unprovability as follows (see \Cref{fig:unprovability-poset} for summary).

\begin{figure}[ht]
\centering
\begin{tikzpicture}[
    node distance=12mm and 20mm,
    theory/.style = {draw, rounded corners=0pt, top color=gray!5, bottom color=gray!5,
                     font=\small, minimum width=25mm, minimum height =8mm,align=center, inner sep=3pt},
    arrow/.style = {-{Latex[length=3mm,width=2mm]}, line width=0.6pt},
    comment/.style = {font=\scriptsize, align=center}
  ]

  \node[theory] (LIND) {$\mathsf{S}^{1}_2 $};
  \node[comment, right=1mm of LIND] {Open problem whether \\ unprovability holds here};

  \node[theory, below left=8mm and 0.5mm of LIND, xshift = 5mm] (BBpoly) {$\PV_1 + \mathsf{BB}(\Sigma^b_{1})$};
  \node[comment, left=1mm of BBpoly] {Best theory known \\ for the unprovability of \\ $\P \subseteq \mathsf{SIZE}[n^k]$};

  \node[theory, below right=8mm and 0.5mm of LIND, xshift = -5mm] (LINDlog) {$\PV_1 + \mathsf{LLIND}(\mathsf{s}\Sigma^b_{1})$};
  \node[comment, right=1mm of LINDlog] {Best theory known \\ for the unprovability of \\ $\mathsf{NSubExp} \not\subseteq \mathsf{Avg-} \mathsf{coNSIZE}[2^{n^\delta}]$};

  \node[theory, below=24mm of LIND] (BBlog) {$\PV_1 + \mathsf{BB}(\mathsf{s}\Sigma^b_{1},\log)$};
  \node[comment, left=1mm of BBlog] {New best theory for \\\textbf{both} unprovability results};

  \node[theory, below=8mm of BBlog] (PV1) {$\PV_1$};
  \node[comment, right=1mm of PV1] {Best unprovability results \\ known until now};

  % Inclusions (arrows)
  \draw[arrow] (PV1) -- (BBlog); 
  \draw[arrow] (BBlog) -- (BBpoly);
  \draw[arrow] (BBlog) -- (LINDlog);
  \draw[arrow] (BBpoly) -- (LIND);
  \draw[arrow] (LINDlog) -- (LIND);

\end{tikzpicture}
\caption{Summary of the unprovability results. Both results were previously known for the theory $\PV_1$, while we show them for $\PV_1+\mathsf{BB}(\Sigma^b_{1})$ and $\PV_1 + \mathsf{LLIND}(\mathsf{s}\Sigma^b_{1})$. We get unprovability of both results for the theory $\PV_1 + \mathsf{BB}(\mathsf{s}\Sigma^b_{1},\log)$ which lies in the intersection. All the inclusions are strict under the hardness assumption $\NP\not\subseteq \P/\poly$.}
\label{fig:unprovability-poset}
\end{figure}
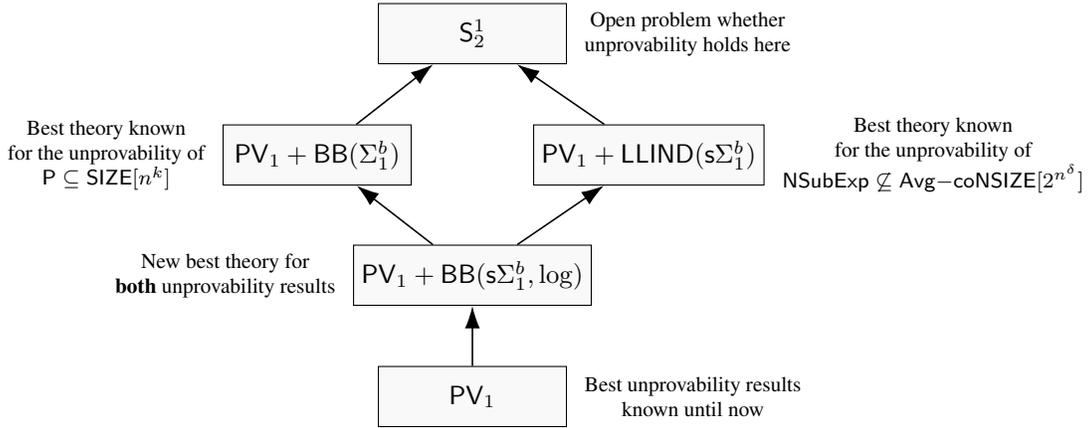

\begin{theorem}[Theorem~\ref{theoremkobb}]\label{thm:krajicekoliveira}
For every $k \in \N$,
        $
        \PV_1+\mathsf{BB}(\Sigma^b_1) \not\vdash \P\subseteq \mathsf{SIZE}[n^k].
        $
\end{theorem}
\begin{theorem}[Theorem~\ref{theorempsllind}]\label{thm:pichsanthanam}
        $
        \PV_1+\mathsf{LLIND}(\mathsf{s}\Sigma^b_1) \nvdash \mathsf{NSubExp} \not\subseteq \mathsf{Avg}\text{-} \mathsf{coNSIZE}[2^{n^\delta}].
        $
\end{theorem}

Here, we use $\mathsf{Avg}$ on the uniform distribution. In other words, the lower bound states that for any $\mathsf{NSubExp}$ Turing machine $M$ and any $D\in \mathsf{coNSIZE}[2^{n^\delta}]$, $\Pr_{x\in \{0,1\}^n}[M(x)=D(x)]\leq 1/2 + 1/2^{n^\delta}$.

It follows from the previous results that the theory $\PV_1+\mathsf{BB}(\mathsf{s}\Sigma^b_1,\log)$ is contained in both of these theories and is still stronger than $\PV_1$ unless $\NP\subseteq\P/\poly$. We thus obtain the following as a corollary.

\begin{theorem}[Corollary \ref{corollarybothunprovabilities}]
\begin{itemize}
    \item[]
    \item For every $k \in \N$, $\PV_1+\mathsf{BB}(\mathsf{s}\Sigma^b_1,\log)\not\vdash \P\subseteq \mathsf{SIZE}[n^k].$
    \item For every rational $\delta\in (0,1)$, $\PV_1+\mathsf{BB}(\mathsf{s}\Sigma^b_1,\log) \nvdash  \mathsf{NSubExp} \not\subseteq \mathsf{Avg-} \mathsf{coNSIZE}[2^{n^\delta}]$.
\end{itemize}
\end{theorem}

\subsection{Open Problems}

All the separations of theories we give are conditioned upon reasonable complexity-theoretic assumptions. The most general open problem would be to extend upon these results by proving them under weaker assumptions, and hopefully to eventually resolve the state of the separations without any assumptions, which seems like a distant goal at this point.

    Is it possible to show that $\mathsf{T}^i_2\not \subseteq\mathsf{S}^i_2+\mathsf{BB}(\Sigma^b_{i+1})$, under some plausible assumptions?
The conditional separation $\mathsf{S}^i_2\subsetneq \mathsf{S}^{i}_2 +\mathsf{BB}(\Sigma^b_{i+1}) \subsetneq \mathsf{S}^{i+1}_2$ is related to another open problem. The inclusions resemble the situation from Buss' hierarchy, where $\mathsf{S}^i_2\subseteq \mathsf{T}^i_2 \subseteq \mathsf{S}^{i+1}_2$. By a result of Ressayre~\cite{ressayre1986conservation}, the theory $\mathsf{S}^{i}_2 +\mathsf{BB}(\Sigma^b_{i+1})$ is $\forall\Sigma^b_{i+1}$-conservative over $\mathsf{S}^i_2$. This contrasts with the expected situation between $\mathsf{S}^i_2$ and $\mathsf{T}^i_2$, where different propositional proof systems correspond to the theories and thus no conservativity is expected. 

Hence arises the question of the difference between $\mathsf{S}^i_2+\mathsf{BB}(\Sigma^b_{i+1})$ and $\mathsf{T}^i_2$. We have already showed that $\mathsf{S}^i_2+\mathsf{BB}(\Sigma^b_{i+1})\not \subseteq \mathsf{T}^i_2$, assuming $\Sigma^p_{i+1}\not\subseteq\Delta^p_{i+1}/\poly$. By~\cite{krajivcek1993fragments} combined with the aforementioned conservativity result of $\mathsf{S}^i_2+\mathsf{BB}(\Sigma^b_{i+1})$ over $\mathsf{S}^i_2$, we obtain that the other separation $\mathsf{T}^i_2\not \subseteq\mathsf{S}^i_2+\mathsf{BB}(\Sigma^b_{i+1})$ is only known to follow from $\P^{\Sigma^p_i}$ being distinct from $\P^{\Sigma^p_i}[O(\log n)]$, the class of problems decidable by a polynomial-time oracle machine with $O(\log n)$-many calls to a $\Sigma^p_i$-oracle, and this hypothesis is not known to be implied by the polynomial hierarchy not collapsing. Therefore, it is natural to ask the following.

\begin{problem}
    Is it possible to show that $\mathsf{T}^i_2\not \subseteq\mathsf{S}^i_2+\mathsf{BB}(\Sigma^b_{i+1})$, under the assumption that the polynomial hierarchy does not collapse?
\end{problem}

What is more, in Corollary~\ref{corollarybothunprovabilities} we show that both unprovabilities of Krajíček--Oliveira and Pich--Santhanam are true for any theory which is included in $\PV_1+\mathsf{BB}(\Sigma^b_1)$ and $\PV_1+\mathsf{LLIND}(\mathsf{s}\Sigma^b_1)$. The authors are under the suspicion that $\PV_1+\mathsf{BB}(\mathsf{s}\Sigma^b_1,\log)$ might not be the greatest such theory.

\begin{problem}
    Is it possible to axiomatize the greatest common subtheory of $\PV_1+\mathsf{BB}(\Sigma^b_1)$ and $\PV_1+\mathsf{LLIND}(\mathsf{s}\Sigma^b_1)$ by some natural axiom scheme? Can we separate this theory from $\PV_1+\mathsf{BB}(\mathsf{s}\Sigma^b_1,\log)$ under some plausible assumptions?
\end{problem}

\subsection{Organisation}

The rest of the paper is organised as follows:
\begin{enumerate}
    \item In \Cref{sec:preliminaries}, we state some preliminaries on Bounded Arithmetic and the different axioms we use.
    \item In \Cref{sec:stgames}, we prove \Cref{thm:st1,thm:st2}, which concern the main technical tool of the separations among different classes of Student-Teacher Games.
    \item In \Cref{sec:witnessing}, we provide a unifying proof for witnessing theorems, and we show how it applies to the length induction axiom and the bounded replacement axiom, which correspond respectively to rounds and parallel queries per round in Student-Teacher Games.
    \item In \Cref{sec:separations}, we combine the results of the previous sections to prove \Cref{thm:bbsep,thm:lindsep,thm:bblindsep,thm:lindbbsep} about the separations of different Bounded Arithmetic theories.
    \item In \Cref{sec:unprovability}, we use the new witnessing theorems to lift the aforementioned unprovability results to seemingly stronger theories (\Cref{thm:krajicekoliveira,thm:pichsanthanam}).
    \item In \Cref{sec:appendix} of the Appendix, we provide some more detailed calculations for the functions used in the generalised versions of our theorems.
\end{enumerate}

%% file: preliminaries.tex
\section{Preliminaries}\label{sec:preliminaries}

We refer to \cite{buss1985bounded} and \cite{krajicek1995bounded} for an introduction to Bounded Arithmetic. Here, we state some facts from these sources which are useful for the purposes of the paper.

We first need to classify formulas in different classes.
\vspace{-0.5em}
\begin{definition}[Formulas]
\hfill 
\vspace{-0.5em}
\begin{itemize}
    \item The quantifiers of the form $\exists x\leq \abs{t}$ or $\forall y\leq \abs{t}$ for some term $t$, are called \emph{sharply bounded quantifiers}, and the quantifiers of the form $\exists x\leq t$ or $\forall y\leq t$ for some term $t$, are called \emph{bounded quantifiers}. We sometimes use the notation $\exists^b$ and $\forall^b$ for bounded quantifiers, if we do not care to determine the term $t$.
    \item The class of formulas $\Sigma^b_0=\Pi^b_0$ are called sharply bounded formulas, and they are the formulas where the only quantifiers appearing are sharply bounded.
    \item For $i\geq 0$, the classes $\Sigma^b_{i+1}$ and $\Pi^b_{i+1}$ are the smallest formula classes that satisfy the properties:
    \begin{enumerate}
        \item $\Sigma^b_{i}\cup \Pi^b_i\subseteq \Sigma^b_{i+1}$ and $\Sigma^b_{i}\cup \Pi^b_i\subseteq \Pi^b_{i+1}$.
        \item If $\phi\in \Sigma^b_{i+1}$, then $\exists^b\phi\in \Sigma^b_{i+1}$, and $\neg \phi \in \Pi^b_{i+1}$.
        \item If $\phi\in \Pi^b_{i+1}$, then $\forall^b\phi\in \Pi^b_{i+1}$, and $\neg \phi \in \Sigma^b_{i+1}$.
        \item $\Sigma^b_{i+1}$ and $\Pi^b_{i+1}$ are closed under sharply bounded quantification, disjunction and conjunction.
    \end{enumerate}
    \item For $i\geq 0$, the classes \emph{strict} $\Sigma^b_{i+1}$ ($\mathsf{s}\Sigma^b_{i+1}$) are defined as the formulas in the form
    $$
    \exists^b\:\forall^b\:\exists^b\dots \phi,
    $$
    where there are $i+1$ alternating bounded quantifiers and $\phi$ is a sharply bounded formula. The classes \emph{strict} $\Pi^b_{i+1}$ ($\mathsf{s}\Pi^b_{i+1}$) are defined dually.
    \item The class $\Delta^b_{i+1}$ is defined by the $\Sigma^b_{i+1}$ formulas which are proved to be equivalent with a $\Pi^b_{i+1}$ formula over a base theory specified (if not specified, it is implied we use $\SOT$).
\end{itemize}
\end{definition}

We introduce a generalisation of the sharply bounded quantifiers, when we have different functions of bounds.
\begin{definition}[$\bd$-bounded Formulas]
    For a bound function $\bd$, which is in the specified language, such that $\bd(x)\leq x$, the quantifiers $\exists x\leq \bd(|t|)$ and $\forall y\leq \bd(|t|)$ for some term $t$, are called $\bd$-sharply bounded quantifiers. For $i\geq 0$, the classes $\bd\Sigma^b_{i+1}$ and $\bd\Pi^b_{i+1}$ are the closures of $\mathsf{s}\Sigma^b_{i+1}$ and $\mathsf{s}\Pi^b_{i+1}$ over $\bd$-sharply bounded quantification, disjunction and conjunction.
\end{definition}

\begin{definition}[Axioms]
The theories we are interested in can be formalised by different axiom schemes. The main ones are:
\begin{enumerate}
    \item $\mathsf{IND}(\varphi):= \forall x \Big[\big(\varphi(0) \land \forall n <  x\, (\varphi(n)\rightarrow \varphi(n+1))\big) \rightarrow \varphi(x)\Big]$
    \item $\mathsf{LIND}(\varphi):= \forall x \Big[\big(\varphi(0) \land \forall n <  |x|\, (\varphi(n)\rightarrow \varphi(n+1))\big) \rightarrow \varphi(|x|)\Big]$
    \item $\mathsf{PIND}(\varphi):= \forall x \Big[\big(\varphi(0) \land \forall n \leq x\, (\varphi(\lfloor n/2 \rfloor )\rightarrow \varphi(n))\big) \rightarrow \varphi(x)\Big]$
    \item $\mathsf{MIN}(\varphi):= \forall x \Big[\varphi (x) \rightarrow \exists n\leq x\;\forall m<n\;\varphi(n) \land \neg \varphi(m) \Big]$
    \item $\mathsf{LMIN}(\varphi):= \forall x \Big[\varphi (x) \rightarrow \exists n\leq x\;\forall m\leq x\;\varphi(n) \land \big(\abs{m}<\abs{n} \rightarrow \neg \varphi(m) \big)\Big]$
    \item $\mathsf{MAX}(\varphi):= \forall x \Big[\varphi (0) \rightarrow \exists n\leq x\;\forall m\leq x\;\varphi(n) \land \big(n<m \rightarrow \neg \varphi(m) \big)\Big]$
    \item $\mathsf{LMAX}(\varphi):= \forall x \Big[\varphi (x) \rightarrow \exists n\leq x\;\forall m\leq x\;\varphi(n) \land \big(\abs{n}<\abs{m} \rightarrow \neg \varphi(m) \big)\Big]$
    \item $\mathsf{BB}(\varphi):= \forall x,t \Big[(\forall j \leq |x| \, \exists y \leq t \; \varphi(j,y))\;\rightarrow\;(\exists w \, \forall j \leq |x| \, \varphi(j, w_j)) \Big]$
\end{enumerate}
\end{definition}

The theories $\mathsf{S}^i_2$ and $\mathsf{T}^i_2$ of Buss's hierarchy are defined using $\mathsf{BASIC}$, which is a fixed set of axioms that determine basic operations of the symbols and arithmetic, and extending it by different kinds of induction:
$$
\mathsf{S}^i_2:= \mathsf{BASIC} + \mathsf{LIND}(\Sigma^b_i), \quad \quad \quad \mathsf{T}^i_2 := \mathsf{BASIC} + \mathsf{IND}(\Sigma^b_i).
$$
We will also need the theories $\PV_i$, defined in \cite{KrajicekPT91,Coo75}. The theory $\PV_{i+1}$ inductively defines symbols for all $\Sigma^p_i$ predicates and it is closed under definition by cases (see next subsection) and under limited recursion on notation. In fact, we can regard it as a universal theory that contains symbols for all functions in $\FP^{\Sigma^p_i}$ and also contains induction on open formulas. It is easy to see that $\PV_{i+1}$ can be regarded as a fully conservative extension of $\mathsf{T}^i_2$; that is why we use them interchangeably throughout the paper. The important property of $\PV_{i+1}$ is its universality, which enables us to use them for witnessing theorems. Let us note that for $i\geq 0,j\geq 1$ we always consider the classes of formulas $\Sigma^b_i$ and $\Pi^b_i$ to be over the language of $\PV$ and we do not allow $\PV_j$-symbols unless explicitly stated, this is to keep the complexity of the definable sets by the formulas intact.

With the language of PV, we mean the language of the theory $\PV_1$, which contains symbols for all $\FP$ functions. This is why it is useful to work on this language, and from now on, apart from the case of the theories $\PV_{i+1}$ for $i\geq 1$, this is the language we use.

In the following, we see some equivalences of different axioms.

\begin{proposition}\label{prop:axioms}
    Over the base theory $\SOT$ or $\PV_1$, we have the following equivalences for all $i\in \N$:
    \begin{enumerate}
        \item[(a)] $\mathsf{IND}(\Sigma^b_i) \iff \mathsf{MIN}(\Sigma^b_i) \iff \mathsf{MAX}(\Sigma^b_i)$.
        \item[(b)] $\mathsf{LIND}(\Sigma^b_i) \iff \mathsf{PIND}(\Sigma^b_i) \iff \mathsf{LMIN}(\Sigma^b_i) \iff \mathsf{LMAX}(\Sigma^b_i)$.
    \end{enumerate}
\end{proposition}

This means that we can axiomatise the different theories in Buss's hierarchy in the following way:
$$
\mathsf{S}^i_2 := \mathsf{BASIC} + \mathsf{LIND}(\Sigma^b_i) = \mathsf{BASIC} + \mathsf{PIND}(\Sigma^b_i) = \mathsf{BASIC} + \mathsf{LMIN}(\Sigma^b_i) = \mathsf{BASIC} + \mathsf{LMAX}(\Sigma^b_i).
$$
$$
\mathsf{T}^i_2 := \mathsf{BASIC} + \mathsf{IND}(\Sigma^b_i) = \mathsf{BASIC} + \mathsf{MIN}(\Sigma^b_i) = \mathsf{BASIC} + \mathsf{MAX}(\Sigma^b_i).
$$
$\mathsf{BASIC}$ here is a fixed set of axioms that determine basic operations of the symbols and arithmetic.

In this work, we also consider the axioms with some specific bound, $\bd$ on the length of the variables. This function $\bd$ must be a symbol of $\PV$, increasing and  $\PV_1$ must prove that $\bd(x)\leq x$.

\begin{enumerate}
    \item $\mathsf{LIND}(\varphi,\bd):= \forall x \Big[\big(\varphi(0) \land \forall n < \bd( |x|)\, (\varphi(n)\rightarrow \varphi(n+1))\big) \rightarrow \varphi(\bd(|x|))\Big]$
    \item $\mathsf{LMIN}(\varphi,\bd):= \forall x \Big[\varphi (x) \rightarrow \exists n\leq x\;\forall m\leq x\;\varphi(n) \land \big(\bd(\abs{m})<\bd(\abs{n}) \rightarrow \neg \varphi(m) \big)\Big]$
    \item $\mathsf{LMAX}(\varphi,\bd):=\forall x \Big[\varphi (x) \rightarrow \exists n\leq x\;\forall m\leq x\;\varphi(n) \land \big(\bd(\abs{n})<\bd(\abs{m} )\rightarrow \neg \varphi(m) \big)\Big]$
    \item $\mathsf{BB}(\varphi,\bd):=  \forall x,t \Big[(\forall j \leq \bd(|x|) \, \exists y \leq t \; \varphi(j,y))\;\rightarrow\;(\exists w \, \forall j \leq \bd(|x| )\, \varphi(j, w_j)) \Big]$
\end{enumerate}

In a similar way with \Cref{prop:axioms}, we can show the following.

\begin{proposition}\label{prop:lind}
    Over the base theory $\SOT$ or $\PV_1$, we have the following equivalences for all $i\in \N$:
$$\mathsf{LIND}(\Sigma^b_i,\bd)  \iff \mathsf{LMIN}(\Sigma^b_i,\bd) \iff \mathsf{LMAX}(\Sigma^b_i,\bd).$$

\end{proposition}

The bounded replacement axiom scheme for $\Sigma^b_{i+1}$ formulas is used in a theory to convert a $\Sigma^b_{i+1}$ formula with arbitrary sharply bounded quantifiers into an equivalent one which is strict. This is sometimes called ``sharply bounded quantifier exchange property'', analogously with the property in unbounded arithmetic. The following is similar as Corollary 15 in \cite{buss1985bounded}.

\begin{proposition}[Sharply bounded quantifier exchange property]
    The theory $\PV_1+\mathsf{BB}(\mathsf{s}\Sigma^b_{i+1})$ proves that for every $\Sigma^b_{i+1}$ formula, $\varphi$, there is an equivalent $\mathsf{s}\Sigma^b_{i+1}$ formula, $\varphi'$.
\end{proposition}

We generalise this property to $\bd$-sharply bounded quantifiers, which can be considered as the goal of defining $\mathsf{BB}(\mathsf{s}\Sigma^b_{i+1},\bd)$ in the first place.

\begin{proposition}[$\bd$-sharply bounded quantifier exchange property]\label{prop:bd-exchange}
    The theory $\PV_1+\mathsf{BB}(\mathsf{s}\Sigma^b_{i+1},\bd)$ proves that for every $\bd\Sigma^b_{i+1}$ formula, $\varphi$, there is an equivalent $s\Sigma^b_{i+1}$ formula, $\varphi'$.
\end{proposition}

\begin{proof}
    We prove this inductively, and we only consider the case of $\bd$-sharply bounded quantification, the cases of disjunction and conjunction are trivial. Assume that $\varphi = (\forall x\leq \bd(t))\, \varphi_0$, where $\varphi_0\in \bd\Sigma^b_{i+1}$ is equivalent over the theory with the $\mathsf{s}\Sigma^b_{i+1}$ formula $\varphi_0'$. We can write $\varphi_0'=\exists^by \:\varphi_1'$, so from $\mathsf{BB}(\varphi_1',\bd)$, we get
    $$
    \forall x\leq \bd(t) \;\exists^b y\;\varphi_1'(x,y) \rightarrow \exists^b w\;\forall x\leq \bd(t)\;\varphi_1'(x,w_x)
    $$
    This is actually an equivalence and the left-hand side is equivalent with $\varphi$, while the right-hand side is strict, which gives us the desired result.
\end{proof}

An important fact for the theories we use, is that length induction can prove the bounded replacement scheme for the same class of formulas. The original result is the following.

\begin{theorem}[\cite{buss1985bounded},Theorem 14]\label{theoremsotprovesbb1}
    For $i\geq 0$, $\mathsf{S}^{i+1}_2 \vdash \mathsf{BB}(\Sigma^b_{i+1})$.
\end{theorem}

In the generalised form, we have:

\begin{theorem}[Length induction proves bounded replacement]\label{theoremllindprovesbb}
    For $i\geq 0$, and a function $\bd$, as above $\mathsf{T}^{i}_2 + \mathsf{LIND}(\mathsf{s}\Sigma^b_{i+1},\bd) \vdash \mathsf{BB}(\bd\Sigma^b_{i+1},\bd)$.
\end{theorem}

\begin{proof}
    Let $\varphi\in \mathsf{s}\Sigma^b_{i+1}$. We want to show that $\mathsf{LIND}(\mathsf{s}\Sigma^b_{i+1},\bd) \vdash \mathsf{BB}(\varphi,\bd)$. Assume that $\forall x\leq \bd(\abs{t}) \;\exists y\leq s \; \varphi(x,y)$.

We write $\varphi = \exists^b z \varphi'(x,y,z_1,z_2)$, where $\varphi'$ is strict $\Pi^b_i$ and we consider the $\mathsf{s}\Sigma^b_{i+1}$ formula:
$$
\psi(u):= \exists^b w, z' \;\forall^b x \; \Big(\mathsf{Seq}(w) \land \mathsf{Seq}(z') \land \mathsf{Len}(w)= \mathsf{Len}(z')=u \land (x\leq u \rightarrow \varphi'(x,w_x,z'_x))\Big).
$$
Notice that this is truly strict, since the universal quantifier $\forall^b x $ can be merged with the first quantifier of $\varphi'$.

We want to show that $\forall u< \bd(\abs{t}) \;\psi(u)\rightarrow \psi(u+1)$. Assume $\psi(u)$, this means that we have two sequence of length $u$, let $w_u$, $z'_u$, where for each $x\leq u$, $\varphi'(x,w_x,z_x)$. By the assumption, we know that $\exists y\leq s\; \exists^b z \;\varphi(u+1,y,z)$, since $u+1\leq \bd(\abs{t})$. We name those $y_{u+1}$ and $z_{u+1}$ and we consider the sequences $\mathsf{Ext}(w_u,y_{u+1})$ and $\mathsf{Ext}(z'_u,z_{u+1})$, where the function $\mathsf{Ext}(w,x)$ extends the sequence $w$ by putting the element $x$ in the end. It is easy to see that these satisfy the existential quantifiers of $\psi(u+1)$.

By $\mathsf{LIND}(\psi,\bd)$, we get that 
$$
\exists^b w, z' \;\forall^b x\leq \bd(\abs{t}) \; \varphi'(x,w_x,z'_x).
$$
By changing the two last quantifiers, we prove that the replacement holds for $\varphi$. Hence, $\mathsf{T}^{i}_2 + \mathsf{LIND}(\mathsf{s}\Sigma^b_{i+1},\bd) \vdash \mathsf{BB}(\mathsf{s}\Sigma^b_{i+1},\bd)$ However, by \Cref{prop:bd-exchange}, $\mathsf{T}^{i}_2 + \mathsf{BB}(\mathsf{s}\Sigma^b_{i+1},\bd) \vdash \mathsf{BB}(\bd\Sigma^b_{i+1},\bd)$, thus we get the desired result.
\end{proof}

The last result we recall, shows that the replacement axiom on $\Sigma^b_{i+1}$ formulas can prove length induction on $\Sigma^b_i$ formulas. Later, we show that it cannot prove length induction on $\Sigma^b_{i+1}$, even if it is weakened by some sublinear bound $\bd$. It is still open whether, it also proves induction on $\Sigma^b_i$ formulas.

\begin{theorem}[\cite{buss1985bounded},Theorem 16]\label{theorembbproveslind}
    For $i\geq 1$, $\SOT + \mathsf{BB}(\Sigma^b_{i+1}) \vdash \mathsf{S}^i_2$.
\end{theorem}

%% file: student-teacher.tex
\section{Student-Teacher Games}\label{sec:stgames}

In this section, we prove the two main theorems about Student-Teacher Games, where the former shows that adaptivity gives more computational power to the Student (\Cref{thm:st1}), and the latter states the same for parallelism (\Cref{thm:st2}).

As a note, for the proof of the theorems, we will need to make a conversion between different input sizes; the first input is for the Student-Teacher Games search problem, and the other is the input for a decision problem about a language $L$. We use $m$ for the input size of the former, and $n$ for the input size of the latter. In the two formulas that we use for the separations taking place in \Cref{thm:st1,thm:st2}, we consider $p(n)$ different instances of the decision problem, $x_i$, each of size $n$. This means that the total size of all these strings, which is actually the input for the Student-Teacher Game search problem, is $m=n\cdot p(n)$. We want the function $p(n)$ to be determined according to our needs for the necessary separation. This separation will be achieved with the difference of $p(n)$ and $p(n)-1$, but we want to convert this result as some function of $m$. Suppose we are given the function $g(m)$ with $g(m)\leq m$, which may correspond to the number of rounds, or queries, etc., and we want the separation between $g(m)$ and $g(m)-1$. This means that we should choose the appropriate $p(n)$, such that $p(n)=g(m)$ (which equivalently means that $p(n)=g(n\cdot p(n))$. This task is non-trivial, but we show that the function $p(n)$ exists in \Cref{sec:input}. From now on, in the proofs, we use a simpler form of the solution, given by the relations $n=\frac{m}{g(m)}$ and $p(n)=g(m)$, but assuming the process we described.

\subsection{Adaptivity}\label{sec:st1}
       \begin{proof}[Proof of \Cref{thm:st1}]
        This theorem shows the power of adaptivity in Student-Teacher Games. We will show that for any sublinear, unbounded, increasing, polynomial-time function $r(m)$ and any unbounded, increasing, polynomial-time function $1\leq q(m)\leq \poly(m)$,
        $$
        \mathcal{ST}^{\Sigma^p_i}[r(m)+1,1]\subseteq \mathcal{ST}^{\Sigma^p_i}[r(m),q(m)] \implies \Sigma^p_{i+1}\subseteq \Delta^p_{i+1}/\poly.
        $$
        
        Fix $n=\frac{m}{r(m)}$ and $p(n)=r(m)$, as discussed above, and consider an arbitrary $\Sigma^p_{i+1}$ language, $L$, which has the form $\exists \,\varphi $, where $\varphi$ is a $\Delta^p_{i+1}$ predicate ($x\in L\iff \exists y\;\varphi(x,y)$). We define the new predicate $\varphi'(x,y,z):= \varphi(x,z)\rightarrow \varphi(x,y)$, also in $\Delta^p_{i+1}$, and take the sentence
    \begin{equation}
        \Phi:=\forall x_1,\dots,x_{p(n)}\in \{0,1\}^n\; \exists y_1,\dots,y_{p(n)} \;\forall i,z_1,\dots,z_{p(n)}\; \varphi'(x_i,y_i,z_i),
        \nonumber
    \end{equation}
    which defines the following $\mathsf{TF\Sigma^p_{i+2}}$ search problem:

    \emph{``Given $p(n)$ $n$-bit strings, find witnesses (of the first existential quantifier) for all those of them that belong to the language $L$.''}
   
    The problem is characterized by the relation $R(x,y):=\forall i,z_1,\dots,z_{p(n)}\; \varphi'(x_i,y_i,z_i)\in \Pi^p_{i+1}$, where $x=(x_1,\dots,x_{p(n)})$ is the input of size $n\cdot p(n)=m$ and $y=(y_1,\dots,y_{p(n)})$ is the sequence of witnesses.
    
    It is easy to see that the problem can be solved with a Student-Teacher Game of $p(n)+1=r(m)+1$ rounds with one query per round: The Student at every round outputs the witnesses they know along with $0$'s for the witnesses of the instances that they do not have knowledge of, while the Teacher must give at least one correct witness for one of the instances at every round as a counterexample. After at most $p(n)$ rounds, the Student knows all the witnesses, which they output at the last round.

    Assume now that $R(x,y)\in \mathcal{ST}^{\Sigma^p_i}[r(m),q(m)]$, which means that it is solvable by a Student in $p(n)$ rounds with polynomially many queries at each round. Let $f\in \FP^{\Sigma^p_i}$ be the function of the Student, which takes as input the sequence $x_1,\dots,x_{p(n)}$, the number of round $j$ and depending on $j$, $j-1$ counterexamples in the form of $i,z_1,\dots,z_{p(n)}$, and outputs at most $\poly(n)$ different sequences of candidate witnesses. If the number of the round is greater than $p(n)$, then the function rejects. The goal is to find a $\poly(n)$-time algorithm that uses this function and $\poly(n)$-size advice in order to decide the $n$-size instances of $L$. This means that $L\in \Delta^p_{i+1}/\poly$.

    Before we continue, we fix the Teacher's strategy. At every round, the Teacher receives $\poly(n)$ number of queries in the form $y^{j}=(y^{j}_{1},\dots,y^{j}_{p(n)})$, where $j$ indices over the queries. Then, the Teacher outputs as counterexample for all the queries, the minimum index $i$, such that $x_i\in L$, but this is not witnessed by any $y^{j}_i$ for all $j$, along with the least lexicographically true witness $z$ of $x_i$. Note that if there is no such minimum, then the Student could have already won, because they could output the tuple with all the correct instances, which can be validated by their computational power. Hence, we disregard these cases.
    
    In order to construct the advice, we start by the set $U$ of all $x\in L$ of size $n$. Formally, $$U:=\{x\;:\; \abs{x}=n \;\land \; \exists y \;\varphi(x,y)\}$$ 
    Let $S$ be a subset of $U$ with $p(n)-1$ elements and $x\in U\setminus S$. We can order the set $S\cup \{x\}$ (e.g. in increasing order) and make the tuple $X_{(S,x)}= (x_1,\dots,x_{l-1}, x, x_l,\dots,x_{p(n)-1})$, where $x_i\in S$ for all $i\in [p(n)-1]$.\footnote{These should not be confused with the $x_i$'s of the formula $\Phi$.} We say that the pair $(S,x)$ is good if providing as input $X_{(S,x)}$ to the Student-Teacher game with the given Student, the fixed Teacher never outputs the index $l$ and a witness for $x$. 

    \begin{claim}
        For every $p(n)$-element subset of $U$, there exists at least a good pair.
    \end{claim}
    To prove the claim, we can order the elements of the set and provide the ordered tuple as input to the Student-Teacher computation. Since the Teacher strategy is to output a single index with the corresponding witness at every round, and the Teacher provides $p(n)-1$ counterexamples, there is at least an index $l$ and a witness $z$, which the Teacher never outputs. The $l$th element of the tuple with the set of the remaining elements form the good pair.
    
    As a result, if the set $U$ has $N$ elements, then there are at least as many good pairs in $U$ as $p(n)$-element sets, which is $\binom{N}{p(n)}$.
    On the other hand there are $\binom{N}{p(n)-1}$-many $(p(n)-1)$-element sets, hence there exists such a set $S$ contained in at least 
    \[
    \binom{N}{p(n)}\binom{N}{p(n)-1}^{-1} = \frac{N-p(n)+1}{p(n)}
    \]
    good pairs.

    We recursively construct the sets $U_i,S_i$ starting with $U_0=U$, as follows:
    If the set $U_i$ has $N_i$ elements, then we take a $(p(n)-1)$-subset $S_i\subseteq U_i$ which is contained in at least $\frac{N_i-p(n)+1}{p(n)}$ good pairs. Afterwards, we construct $U_{i+1}:=U_i\setminus\{x\,:\,(S_i,x) \text{ is a good pair}\}$.

    This means that for $\lambda=\frac{p(n)-1}{p(n)}$, $N_{i+1}=\lambda N_i+\lambda<\lambda N_i+1$, thus $N_k<\lambda^k N_0+k$. For $k=\frac{\log(N)}{\log (\lambda^{-1})}=O(n)$, the last relation implies that $N_k<k$. Therefore, the set $Q=S_0\cup\dots\cup S_{k-1}\cup U_k\subseteq U$ has size $O(n\cdot p(n))$.

    Finally, we describe the non-uniform algorithm that decides $L$ on inputs of size $n$. The correctness of the algorithm comes from the construction of the set $Q$, where all $x\in U$ form a good pair with some $S_i$ or belong to $Q$.

\begin{algorithm}[H]
     \SetKwInOut{Input}{Input}
     \SetKwInOut{Advice}{Advice}
     \caption{The pseudocode of the $\Delta^p_{i+1}/\poly$-algorithm that decides $L$ on input size $n$.}
     \Input{A string $x\in \{0,1\}^{n}$.}
     \Advice{The set $Q$ and the set of the least lexicographic witness for each element of $Q$.}
     \vspace{0.3em}
     If $x\in Q$, then check the corresponding witness and Accept\;
     Otherwise, for all $S_i$, simulate the Student-Teacher Game with input $X_{(S_i,x)}$ using the function $f$ and the witnesses of $S_i$\;
     If any of these Student-Teacher Games generates a valid witness, then Accept\;
     If none of these Games generates a valid witness, then Reject.\\
\end{algorithm}
\end{proof}

\subsection{Parallelism}\label{sec:st2}
     \begin{proof}[Proof of \Cref{thm:st2}]
     This theorem shows the power of the number of parallel queries per round in Student-Teacher Games. We will show that for any sublinear, unbounded, increasing, polynomial-time functions $r_1(m),q_1(m),r_2(m),q_2(m)$ with $r_1(m)q_1(m)<r_2(m)q_2(m)$,
        $$
        \mathcal{ST}^{\Sigma^p_i}[1+r_2(m),q_2(m)]\subseteq \mathcal{ST}^{\Sigma^p_i}[1+r_1(m),q_1(m)] \implies \Sigma^p_{i+1}\subseteq \Delta^p_{i+1}/\poly.
        $$
        
        We fix $n=\frac{m}{r_2(m)q_2(m)}$ and $p(n)=r_2(m)q_2(m)$, as discussed above, and we consider an arbitrary $\Sigma^p_{i+1}$ language, $L$, which has the form $\exists \,\varphi $, where $\varphi$ is a $\Delta^p_{i+1}$ predicate. Now, we take the sentence
    \begin{equation}
        \Psi:=\forall x_1,\dots,x_{p(n)}\in \{0,1\}^n\; \left( \exists i\leq p(n) \; \forall y \;\neg \varphi(x_i,y)\right)\; \lor \; \left( \exists y_1,\dots,y_{p(n)} \; \forall j\leq p(n)\; \varphi(x_j,y_j) \right)
        \nonumber
    \end{equation}
    which defines the following $\mathsf{TF\Sigma^p_{i+2}}$ search problem:

    \emph{``Given $p(n)$ $n$-bit strings, find an index of an element which is not in the language $L$ or find a $p(n)$-long sequence that includes witnesses (of the first existential quantifier) for all of the elements in the input.''}
   
    The problem is characterized by the relation $R(x,y):=\forall i,z_1,\dots,z_{p(n)}\; \varphi'(x_i,y_i,z_i)\in \Pi^p_{i+1}$, where $x=(x_1,\dots,x_{p(n)})$ is the input of size $n\cdot p(n)=m$ and the search terms are $i$ and the sequence $(y_1,\dots,y_{p(n)})$.
    
    It is easy to see that the problem can be solved with a Student-Teacher Game of $r_2(m)+1$ rounds with $q_2(m)$ queries per round: The Student splits the numbers of $[p(n)]$ into $r_2(m)$ sets of $q_2(m)$ elements, and at each round guesses that the index where there is no witness, is one of the elements of the set. They achieve that by the $q_2(m)$ parallel queries. After $r_2(m)$ rounds, the Student has either found an element without a witness or has been provided witnesses for all the elements, which they combine in a sequence and output in the last round.

    Assume now that $R(x,y)\in \mathcal{ST}^{\Sigma^p_i}[1+r_1(m),q_1(m)]$, which means that it is solvable by a Student that has received at most $r_1(m)q_1(m)$ counterexamples from the Teacher (at the last round the Student should be successful). From the given inequality, this means that the Student has received less than $p(n)$ counterexamples. Let $f\in \FP^{\Sigma^p_i}$ be the function of the Student, which takes as input the sequence $x_1,\dots,x_{p(n)}$, the number of round $j$ and depending on $j$, $j-1$ counterexamples in the form of $y,j$. If the number of the round is greater than $p(n)$, then the function rejects. The goal is, as before, to find a $\poly(n)$-time algorithm that uses this function and $\poly(n)$-size advice in order to decide the $n$-size instances of $L$, so that we have $L\in \Delta^p_{i+1}/\poly$.

    Assume now that $R(x,y)\in \mathcal{ST}^{\Sigma^p_i}[r(m),q(m)]$, which means that it is solvable by a Student in $p(n)$ rounds with polynomially many queries at each round. Let $f\in \FP^{\Sigma^p_i}$ be the function of the Student, which takes as input the sequence $x_1,\dots,x_{p(n)}$, the number of round $j$ and depending on $j$, $j-1$ counterexamples in the form of $i,z_1,\dots,z_{p(n)}$, and outputs at most $\poly(n)$ different sequences of candidate witnesses. If the number of the round is greater than $p(n)$, then the function rejects. The goal is to find a $\poly(n)$-time algorithm that uses this function and $\poly(n)$-size advice in order to decide the $n$-size instances of $L$. This means that $L\in \Delta^p_{i+1}/\poly$.

    Again, we fix the Teacher's strategy. For the first part of the disjunction, if the Student has failed, which means that $x_i\in L$, the Teacher provides the least lexicographically witness for $x_i$. We can actually disregard the second disjunct, since the universal quantifier, which the Teacher is supposed to give counterexamples for, is sharply bounded, thus it can be computed efficiently by the Student. This means that in this case the Teacher provides no new information to the Student.

    As in the previous proof, we start by the set $U$ of all $x\in L$ of size $n$, and for a subset $S\subseteq U$ with $p(n)-1$ elements and an element $x\in U\setminus S$, we form the ordered $p(n)$-tuple $X_{(S,x)}$, where $x$ is in the $l$th position. Then, the pair $(S,x)$ is good if providing $X_{(S,x)}$ as input to the Student-Teacher game with the given Student, the Teacher is never queried and never provides witness for the index $l$.

    It is easy to see that every $p(n)$-element subset of $U$ has at least a good pair, because if we order its elements and give it as input to the Student-Teacher game with $1+r_1(m)$ rounds and $q_1(m)$ queries per round, then there must be an index which is never queried, since the number of queries is $r_1(m)q_1(m)<p(n)$.

    The remaining proof follows exactly the proof of \Cref{thm:st1}, in order to construct a $\Delta^p_{i+1}/\poly$ algorithm for deciding $L$.
\end{proof}

%% file: witnessing.tex
\section{Witnessing Theorems}\label{sec:witnessing}

In this section, we prove the witnessing theorems needed for our separations. We start by giving a general witnessing theorem, which relates a theory $S$ with a suitable relativized universal theory $T$ such that the terms of $T$ describe the computational model for the student in the related Student-Teacher game.

\subsection{General Witnessing Theorem}

\begin{definition}
    Assume there are $L$-structures $M,N\models T$ such that $M\subseteq N$. We say $N$ is an $\exists$-elementary extension if for every open formula $\varphi(\overline x,\overline y)$ it holds that for every $\overline m\in M$ we have
    \[M\models (\exists \overline y)\varphi(\overline m,\overline y)\iff N\models (\exists \overline y)\varphi(\overline m,\overline y).\]
\end{definition}

\begin{theorem}\cite[Theorems 3.2 and 3.3]{AVIGAD2002219}\label{theoremherbrandsaturated}
    Let $L$ be a language, $T$ be a universal $L$-theory and let $M\models T$. Then $M$ has an $\exists$-elementary extension $\hat M$ such that for every open $L$-formula if
    \[\hat M\models (\forall \overline x)(\exists y)(\varphi(x,y)),\]
    then there are $L$-terms $t_1,\dots,t_k$ such that
    \[\hat M\models (\forall \overline x)((\varphi(x,t_1)\lor \dots \lor \varphi(x,t_k)).\]
\end{theorem}

\begin{definition}
    Let $L$ be a language, $T$ an $L$-theory. We say a theory $T$ has terms closed under definition by cases by open formulas~\footnote{The name `closed under definition by cases by open formulas' comes from~\cite{forcing}.}, if for every open $L$-formula $\varphi(\overline x)$ and $L$-terms $t_1(\overline x),t_2(\overline x)$ there is an $L$-term $t(\overline x)$ such that
    \[T\vdash (\forall \overline x)((\varphi(\overline x)\land t(\overline x)=t_1(\overline x))\lor (\lnot \varphi(\overline x)\land t(\overline x)=t_2(\overline x))).\]
\end{definition}

We use that the languages we will work in are rich enough to combine the different terms in the disjunction in the previous theorem into a single term. The following is a sufficient condition to do this.

\begin{lemma}\label{lemmapvifthenelse}
    Let $L$ be a language extending $\PV$ and $T$ be an $L$-theory extending $\PV_1$. Then $T$ has terms closed by definition by cases by open formulas.
\end{lemma}
\begin{proof}
    Assume that we have an open $L$-formula $\varphi(x,\overline y)$ and $L$-terms $t_1$ and $t_2$. It is not hard to construct a term $f(x,\overline y)$ with the property
    \[\PV_1\vdash (\forall x)(\forall \overline y)(f(x,\overline y)=1 \leftrightarrow \varphi(x,\overline y))\tag{*}\]
    starting with atomic formulas, taking the $\PV$-symbol which checks for equality of two inputs, we can create a term which can check whether an equality of $\PV$-terms is valid. Using $\PV$-symbols for disjunction and negation, we can obtain (*) for arbitrary open formulas by induction on the complexity of $\varphi$.

    The term $t$ can be constructed as 
    \[\text{IfThenElse}(f(x,\overline y),t_1(x,\overline y),t_2(x, \overline y)),\]
    where $\text{IfThenElse}(x,y,z)$ is a $\PV$-symbol which returns $y$ if $x=1$ otherwise it outputs $z$.
\end{proof}

The following theorem is a general theorem for witnessing of $\forall\exists \forall$ consequences of a theory. If one starts with a theory $S$, extends it with a new binary function symbol $\alpha(x,y)$, which represents the counterexample function of the Teacher in the Student-Teacher protocol, to obtain a theory $T'$, and finds a universal sub-theory $T\subseteq T'$ such that $T'$ is $\forall\exists$-conservative over it and $T$ has terms closed by definitions by open formulas, then there will be a single term in the language of $T$ such that it $T$-provably computes the witness for the existentially quantified variable. In our applications of this theorem, the theory $S$ we does not pose any restriction on the concrete behaviour of $\alpha$ and as such we may assume that $\alpha$ behaves as an appropriate teacher for the given Student-Teacher game. 

\begin{definition}
    Let $L_1$,$L_2$ be languages, $T_1$ be an $L_1$-theory and let $T_2$ be an $L_2$-theory, let $\Gamma$ be a class of $L_1$-formulas, we say that $T_2$ is $\Gamma$-conservative over $T_1$ if for every $\varphi\in\Gamma$ which is also an $L_2$-formula we have that $T_1\vdash \varphi$ if and only if $T_2\vdash \varphi$.
\end{definition}

\begin{theorem}\label{theoremwitnessing}
    Let $\alpha$ be a binary function symbol, let $L_w$ be a language containing $\alpha$ and let $L\subseteq L_w\setminus \{\alpha\}$ be another language. Consider a universal $L_w$-theory $T$ with terms closed under definition by cases by open formulas, its $\forall\exists$-conservative extension $T'$ and an $L$-theory $S\subseteq T'$. Let $\psi$ be an open $L$-formula. If 
    \[S\vdash (\forall x)(\exists y)(\forall z)(\psi(x,y,z))\]
    then there is an $L_w$-term $t$ such that
    \[T\vdash (\forall x)(\psi(x,t(x),\alpha(x,t(x)))).\]
\end{theorem}
\begin{proof}
    Assume that $S\vdash (\forall x)(\exists y)(\forall z)(\psi(x,y,z))$. Then we have by Herbrandization on the variable $z$ that
    \[S\vdash (\forall x)(\exists y)(\psi(x,y,\alpha(x,y))),\]
    by $S\subseteq T'$ we also have
    \[T'\vdash (\forall x)(\exists y)(\psi(x,y,\alpha(x,y))),\]
    by the $\forall\exists$-conservativity of $T'$ over $T$ we have
    \[T\vdash (\forall x)(\exists y)(\psi(x,y,\alpha(x,y))),\]
    and since $T$ is a universal theory we obtain by Herband's theorem $L_w$-terms $t_1,\dots,t_k$ such that
    \[T\vdash (\forall x)(\bigvee_{1\leq i\leq k}(\psi(x,t_i(x),\alpha(x,t_i(x)))),\]
    and by the closedness by definition by cases there is a single $L_w$-term $t$ such that 
    \[T\vdash (\forall x)(\psi(x,t(x),\alpha(x,t(x)))).\qedhere\]
\end{proof}

\subsection{Sharply Bounded Replacement Scheme \texorpdfstring{$\mathsf{BB}(\Sigma^b_j,\bd)$}{BB}}

We will fix a $\PV$-symbol $\bd$ such that $\PV_1$ proves it is upper bounded by the identity, that is $\PV_1\vdash \bd(x)\leq x$. We will now introduce a theory which axiomatizes parallel acces to the oracle $\alpha$.

\begin{definition}
    We define the theory $\PV_j'(\parallel, \alpha,\bd)$ as the theory extending $\PV_j$ by a new function symbol $\alpha(x,y)$ which may have more parameters which we will not display, and for every $c\geq 1$, we also add a function $\Read_c$ with new axioms:
    \begin{align*}
        &(\forall l)(\forall x)((\Seq(l) \land \Len(l)\neq 0 \land \Len(l)\leq \bd(\abs{x})^c) \to (\forall i\leq \Len(l))((\Read_c(l,x))_i=\alpha((l)_i^1,(l)_i^2))),\\
        &(\forall l)(\forall x)((\lnot \Seq(l) \lor \Len(l) = 0\lor \Len(l) > \bd(\abs{x})^c)\to \Read_c(l, x)= 0),
    \end{align*}
    where $\Seq(l)$ is the open $\PV$-formula stating that $l$ codes a sequence of numbers, $\Len$ is the function which gives the number of elements in a sequence and $(-)_i$ is the $\PV$-function symbol which returns the $i$-th element of a sequence, and the superscripts $\phantom{}^1$ and $\phantom{}^2$ are the $\PV$-symbols which return the first and the second component of a pair. We will sometimes omit the subscript and the second parameter of $\Read$ if we want to keep them implicit.

    There are no other axioms about $\alpha$, but in the cases we care about it will represent a function oracle whose output length is bounded by a polynomial in the input length, once a specific bound is fixed we will denote by $\PV_j(\parallel, \alpha,\bd)$ a theory $\PV_j'(\parallel, \alpha,\bd)$ with the extra axiom
    \[\alpha(x,y)\leq t(x,y)\]
    where $t$ is any $\PV_j$-symbol. The function $\Read$ then allows one to query polynomially many values of $\alpha$ at once.
\end{definition}

Note that $\PV_j(\parallel,\alpha,\bd)$ is a sub-theory of $\PV_j(\alpha)$, the relativized $\PV_j$, as $\Read_c$ can be identified with a specific $\PV_j(\alpha)$-symbol with unrestricted oracle access. The theory $\PV_j(\parallel,\alpha,\bd)$ only reasons with functions which can access oracle $\alpha$ for constantly many times and each time ask for $\poly(\bd)$-many queries, as the function $\alpha$ cannot be iterated inside a limited recursion on notation rule.

\begin{lemma}\label{lemmaopenbb}
    Let $\alpha$ be a function satisfying $\abs{\alpha(x,y)}\leq \abs{t(x,y)}$ for some $\PV_j$-term $t$. Then for every $c\geq 1$, every term $s$ in the language of~$\PV_j(\parallel, \alpha,\bd)$ and every open formula $\varphi$ we have that \[\PV_j(\parallel,\alpha,\bd)\vdash (\forall i\leq \bd(\abs{x})^c)(\varphi(i,s(i)))\to (\exists w)(\forall i\leq \bd(\abs{x})^c)(\varphi(i,(w)_i).\]
\end{lemma}
\begin{proof}
    We will proceed by induction on the complexity of $s$, where we will recursively assign to every $\PV_j(\parallel,\alpha,\bd)$-term $u$ another term $w^u$ which computes the witness $w$. If $s$ is a $\PV$-symbol, then the theorem is clear.

    Let $s$ be of the form $\alpha(u(i),v(i))$, for a simpler $\PV_j(\parallel,\alpha,\bd)$-term $u$. It is then possible to compute the value for $w$ by the following term:
    \[\Read_c(\text{PairUp}(w^u,w^v)),\]
    where $w^u$ and $w^v$ are the terms obtained from the induction hypothesis, and $\text{PairUp}$ is $\PV$-symbol which obtains two sequences of the same length on the input and outputs a sequence of pairs where the $i$-th component is the pair of the $i$-th components of the input sequenes.

    Let $s$ be of the form $\Read(u(i))$, for a simpler $\PV_j(\parallel,\alpha,\bd)$-term $u$. It is then possible to compute the value for $w$ by the following term:
    \[\text{Regroup}(\Read_{c+d}(\text{Concat}(w^u)),w^u),\]
    where $w^u$ is the term obtained from the induction hypothesis of $i$ being bounded by $\bd(\abs{x})^d$, $\text{Concat(-)}$ concatenates a sequence of sequences into a single sequence and $\text{Regroup}$ is a function which groups subsequences into elements coding the subsequences so that the resulting sequence of sequences is of the same type as the second argument of $\text{Regroup}$.

    It remains to consider $s$ of the form  $g(u(i))$, for a simpler $\PV_j(\parallel,\alpha)$-term $u$ and a $\PV$-symbol $g$. It is then possible to compute the value for $w$ by the following term:
    \[\text{Apply}_g(w^u),\]
    where $\text{Apply}_g$ is a $\PV$-symbol which receives as an input a sequence and outputs a sequence which is obtained by applying $g$ to every element of the input.
\end{proof}

\begin{definition}
    The axiom scheme $\mathsf{BB}(\Sigma^b_{j}(\alpha),\bd)$ consists of the set $\{\mathsf{BB}(\varphi,\bd);\:\varphi\in\Sigma^b_j(\alpha)\},$
    where $\Sigma^b_j(\alpha)$ consists of all formulas which can be obtained by starting with a $\Sigma^b_j$-formula and substituting $\PV_j\cup\{\alpha\}$-terms for any of the free variables.
\end{definition}

\begin{lemma}\label{lemmapvbbconservativity}
    Let $\alpha$ be a function satisfying $\abs{\alpha(x)}\leq k\abs{x}^k$ for some $k$. Then the theory \[\PV_j(\parallel,\alpha,\bd)+\mathsf{BB}(\Sigma^b_{j}(\alpha),\bd)\] is $\forall\exists$-conservative over $\PV_j(\parallel,\alpha,\bd)$ possibly extended by any universal sentences.
\end{lemma}
\begin{proof}
    Let $M$ be an arbitrary model of $\PV_j(\parallel,\alpha,\bd)$. Let $\hat M$ be the extension of $M$ from Theorem~\ref{theoremherbrandsaturated}. We will show that $\hat M\models \mathsf{BB}(\Sigma^b_j(\alpha),\bd)$.

    Assume that for some $\varphi\in \Sigma^b_j$, $c\geq 1$ and a $\PV$-term $\bd$ we have
    \[\hat M \models (\forall i\leq \bd(\abs{x})^c)(\exists w)(\varphi(i,w)),\]
    by Theorem~\ref{theoremherbrandsaturated} there are $\PV_j(\parallel,\alpha,\bd)$-terms $t_1,\dots,t_k$ such that
    \[\hat M \models (\forall i\leq \bd(\abs{x})^c)(\varphi(i,t_1)\lor \varphi(i,t_2)\lor \dots \lor \varphi(i,t_k)),\]
    by $\PV_1\subseteq \PV_j(\parallel,\alpha,\bd)$ and Lemma~\ref{lemmapvifthenelse} we obtain a single $\PV_j(\parallel,\alpha,\bd)$-term $s$ such that \[\hat M \models (\forall i\leq \bd(\abs{x})^c)(\varphi(i,s)).\] 

    By Lemma~\ref{lemmaopenbb} we get that $\hat M \models (\exists w)(\forall i\leq \bd(\abs{x})^{c})(\varphi(x,(w)_i))$. Therefore, $\hat M\models \mathsf{BB}(\Sigma^b_j,\alpha,b)$, and since it is $\exists$-elementary over $M$, which was chosen as an arbitrary model of $\PV_j(\parallel,\alpha)$, we obtain the theorem.
\end{proof}

\begin{theorem} \label{thm:bbwitness}
    Let $\PV_j+\mathsf{BB}(\Sigma^b_{j},\bd)\vdash (\forall x)(\exists y)(\forall z\leq t)(\varphi(x,y,z))$, where $\varphi$ is open and $t$ is a $\PV_j$-term, then the $\mathbf{TF\Sigma^p_{j+1}}$ problem corresponding to $\varphi$ is in $\mathcal{ST}^{\Sigma^p_{j-1}}[O(1),\poly(\bd)]$.
\end{theorem}
\begin{proof}
     We will obtain $\PV_1(\parallel, \alpha,\bd)$ by accepting the axiom $\alpha(x,y)\leq t(x,y)$. By Lemma~\ref{lemmapvifthenelse}, Lemma~\ref{lemmapvbbconservativity} and Theorem~\ref{theoremwitnessing} we obtain that there is a single $\PV_j(\parallel,\alpha,\bd)$ term $s$ such that
    \[\PV_j(\parallel,\alpha,\bd)\vdash (\forall x)(\varphi(x,s(x),\alpha(x,s(x)))),\]
    the value of the term $s$ can be computed by a student which sends constantly many times $\poly(\bd(\abs{x}))$-many answers to the teacher and obtains a counterexample whenever one exists for each of the answers. The student runs in polynomial-time with $\Sigma^p_{j-1}$ oracle access, as all of the values are computed by $\PV_j$-terms and $\PV_j(\parallel,\alpha,\bd)$ terms where the length of $\alpha$ is polynomially bounded.

    Since there is no fixed interpretation for $\alpha$ except for a bound on the length of its outputs, we can interpret $\alpha$ as any function satisfying
    \[(\exists z)\lnot \varphi(x,s(x),z) \to \lnot \varphi(x,s(x),\alpha(x,s(x))),\]
    that is $\alpha$ is a function which is a teacher correcting every wrong answer. Now any oracle call (or parallel oracle call using $\Read_{d}$ for some $d\geq 1$) can either by answered with a counter-example, or $\alpha$ `admits' that it is correct.
\end{proof}

\subsection{Length \texorpdfstring{$\bd$}{b}-induction \texorpdfstring{$\mathsf{LIND}(\Sigma^b_j,\bd)$}{LIND}}

We will fix a $\PV$-symbol $\bd$ such that $\PV_1$ it is upper bounded by the identity,, that is $\PV_1\vdash \bd(x)\leq x$. 

\begin{definition}
    We define the theory $\PV_j'(\to, \alpha,\bd)$ as the theory extending $\PV_j$ by a new function symbol $\alpha(x,y)$ which may have more parameters we will not display and then proceeding in countably many steps, as described in the following. The construction is analogous to the definition of the equational theory $\PV$.

    The theory $\PV_j'(\to,\alpha,\bd)_0$ is in the language $\PV_j\cup\{\alpha\}$ and contains only the axioms of $\PV_j$. Assume, we have defined the theory $\PV_j'(\to,\alpha,\bd)_k$ for some $k$, we will define $\PV_j'(\to,\alpha,\bd)_{k+1}$ as follows. Consider function symbols $g,h_0,h_1,l_0,l_1$ from the language of $\PV_j'(\to,\alpha,\bd)_{k}$ for which we have that for $i\in\{0,1\}$:
    \[\PV_j'(\to,\alpha,\bd)_k\vdash \abs{h_i(\overline x,y,z)}'\leq\abs{z}+\abs{l_i(\overline x,y)},\]
    then for every $c\geq 1$ a new function symbol $f_c$ is added to the language of $\PV_j(\to,\alpha,\bd)_k$ and along with it a new axiom 
    \begin{align*}
        (\forall x)(\forall y)((\abs{y}\leq \bd(\abs{x})^c\to &(f_c(x,0)=g(x,y))\land \\
        &(f_c(x,2\cdot y)=h_0(x,y,f_c(x,y)))\land\\
        &(f_c(x,2\cdot y+1)=h_1( x,y,f_c(x,y))))\land\\
        \abs{y}>\bd(\abs{x})^c\to &(f_c(x,y)=0)),
    \end{align*}
    we say such an $f_c$ is defined using limited recursion on notation for $\poly(\bd)$-many steps. For each formula in the new language $\varphi(x)$ we also include the length $\bd$-induction axiom, that is for every $c$ we accept
    \[\lnot \varphi(0)\lor \varphi(\bd(\abs{x})^c) \lor (\varphi(h(\overline x))\land \lnot \varphi(h(\overline x)+1)\land h(\overline x)+1\leq \bd(\abs{x})^c),\]
    where $h$ is defined depending on $\varphi$ and $c$ using limited recursion on notation for $\poly(\bd)$-many steps, which emulates searching on the interval from $0$ to $\bd(\abs{x})^c$, going from one end step by step to the other, until one of the  disjuncts is satisfied.
    We define $\PV_j'(\to,\alpha,\bd)=\bigcup_{k\geq 0}\PV_j'(\to,\alpha,\bd)_k.$

    There are no other axioms about $\alpha$, but in the cases we care about it will represent a function oracle whose output length is bounded by a polynomial in the input length, once a specific bound is fixed we will denote by $\PV_1(\to, \alpha,\bd)$ a theory $\PV_1'(\to, \alpha,\bd)$ with the extra axiom
    \[\alpha(x,y)\leq t(x,y)\]
    where $t$ is a $\PV_j$-term which is implicitly chosen to reflect the upper bound on the growth-rate of $\alpha$.
\end{definition}

Note that $\PV_j(\to,\alpha,\bd)$ is a sub-theory of $\PV_j(\alpha)$, which contains only those symbols with at most $\poly(\bd)$-many queries to $\alpha$, as the function $\alpha$ cannot be iterated inside a limited recursion on notation rule more than $\poly(\bd)$-many times.

\begin{definition}
    The axiom scheme $\mathsf{LIND}(\mathsf{s}\Sigma^b_{j}(\alpha),\bd)$ consists of the set $\{\mathsf{LIND}(\varphi,\bd);\:\varphi\in\mathsf{s}\Sigma^b_j(\alpha)\},$
    where $\mathsf{s}\Sigma^b_j(\alpha)$ consists of all formulas which can be obtained by starting with a $\mathsf{s}\Sigma^b_j$-formula and substituting some $\PV_j\cup\{\alpha\}$-terms for the free variables.
\end{definition}

\begin{lemma}\label{lemmapvllind}
    Let $\alpha$ be a function satisfying $\abs{\alpha(x,y)}\leq \abs{t(x,y)}$ for some $\PV_{j+1}$-term $t$. Then the theory \[\PV_j(\to,\alpha,\bd)+\mathsf{LIND}(\mathsf{s}\Sigma^b_{j+1},\alpha,\bd)\] is $\forall\exists$-conservative over $\PV_j(\to,\alpha,\bd)$ possibly extended by any set of true universal sentences.
\end{lemma}
\begin{proof}
    Let $M$ be an arbitrary model of $\PV_j(\to,\alpha,\bd)$. Let  $\hat M$ be the extension of $M$ from Theorem~\ref{theoremherbrandsaturated}. We will show that $\hat M\models \mathsf{LIND}(\mathsf{s}\Sigma^b_j,\alpha,\bd)$.

    Assume, that for some open $\PV_j$-formula $\varphi(i,w)$ possibly with extra parameters and a $\PV$-term $t$ we have
    \[\hat M \models (\exists w\leq t)\varphi(0,w) \land (\forall i\leq \bd(\abs{x})^c)((\exists w)(\varphi(i,w))\to (\exists w)(\varphi(i+1,w))),\]
    by Theorem~\ref{theoremherbrandsaturated} there are $\PV_1(\to,\alpha,\bd)$-terms $t_1,\dots,t_k$ such that
    \[\hat M \models (\forall i\leq \bd(\abs{x})^c)(\forall w\leq t)(\bigvee_{1\leq i\leq k}((\varphi(i,w))\to (\varphi(i+1,t_i)))),\]
    and it is possible to only obtain a single $\PV_j(\to,\alpha,\bd)$-term $s$ such that 
    \[\hat M \models (\forall i\leq \bd(\abs{x})^c)(\forall w\leq t)((\varphi(i,w))\to (\varphi(i+1,s))),\]
    by limited recursion on notation on $s$ for $\poly(\bd)$-many steps it is possible to obtain a function $s'$ such that $\PV_j(\to,\alpha,\bd)$ proves
    \[(\forall \overline x)(\forall w)(\forall y, \abs{y}\leq \bd(\abs{x})^c)(s'(w,0) =s(w,0)\land (\forall r\leq 1)(s'(w,2\cdot y+r)=s(s'(w,y),y))),\]
    and thus by length $\bd$-induction on $\abs{y}$ inside $\hat M$ we obtain
    \[\hat M \models (\forall w\leq t)(\varphi(0,w)\to \varphi(\bd(\abs{x})^c,s'(w,\bd(\abs{x})^c)),\]
    and thus
    \[\hat M \models (\exists w)\varphi(\bd(\abs{x})^c,w).\]

    Therefore, $\hat M\models \mathsf{LIND}(s\Sigma^b_j,\alpha,\bd)$, and since it is $\exists$-elementary over $M$, which was chosen as an arbitrary model of $\PV_j(\to,\alpha,\bd)$, we obtain the theorem.
\end{proof}

\begin{theorem}\label{thm:lindwitness}
    Let $\PV_j+\mathsf{LIND}(\Sigma^b_j,\bd)\vdash (\forall x)(\exists y)(\forall z\leq t)(\varphi(x,y,z))$, where $\varphi$ is open and $t$ is an $\PV_j$-term. Then the $\mathsf{TF\Sigma^p_{j+1}}$ problem corresponding to $\varphi$ is in $\mathcal{ST}^{\Sigma^b_{j-1}}[\poly(\bd),1]$. 
\end{theorem}
\begin{proof}
    We use the bound $\alpha(x,y)\leq t(x,y)$ to obtain the theory $\PV_j(\to,\alpha,\bd)$. Now, by Lemma~\ref{lemmapvifthenelse}, Lemma~\ref{lemmapvllind} and Theorem~\ref{theoremwitnessing} we obtain that there is a $\PV_j(\to,\alpha,b)$-term $s$ such that
    \[\PV_j(\to,\alpha,b)\vdash (\forall x)(\varphi(x,s(x),\alpha(x,s(x)))),\]
    the value of the term $s$ can be computed by a student which sends $\poly(\bd)$-many times a single answer to the teacher and obtains a counterexample whenever one exists. Since all of the terms are either $\PV_j$-terms or $\PV_j(\to,\alpha,\bd)$ terms which are bounded by $\PV_j$-terms the student runs in polynomial-time in the length of $x$ with $\Sigma^p_{j-1}$-oracle access.

    As in the previous section, there is no fixed interpretation for $\alpha$ except for a bound on the length of its outputs and we can interpret $\alpha$ as any function satisfying
    \[(\exists z)\lnot \varphi(x,s(x),z) \to \lnot \varphi(x,s(x),\alpha(x,s(x))),\]
    that is $\alpha$ is a function which is a teacher correcting every wrong answer, then any oracle call can either by answered with a counter-example, or $\alpha$ `admits' that it is correct.
\end{proof}

%% file: separations.tex
\section{Separations of Theories}\label{sec:separations}

We will now use the \Cref{thm:st1,thm:st2} and their proofs to show the separations in \Cref{thm:bbsep,thm:lindsep,thm:lindbbsep}.

First, we revisit the sentences $\Phi$ and $\Psi$ from the proofs of the \Cref{thm:st1,thm:st2}. We will parametrise them by the number of elements in the sequence of the input.
$$
\Phi_{p}:= \forall n\in \mathsf{Log}\; \forall x_1,\dots,x_{p(n)}\in \{0,1\}^n\; \exists y_1,\dots,y_{p(n)} \;\forall i,z_1,\dots,z_{p(n)}\; \varphi'(x_i,y_i,z_i)
$$
\begin{align*}
    \Psi_{p}:= \forall n\in \mathsf{Log}\; \forall x_1,\dots,x_{p(n)}\in \{0,1\}^n\; &\left( \exists i\leq p(n) \; \forall y \;\neg \varphi(x_i,y)\right)\; 
\lor\\ &\left( \exists y_1,\dots,y_{p(n)} \; \forall j\leq p(n)\; \varphi(x_j,y_j) \right)
\end{align*}

To write these in the language of $\PV$, we can change the $n\in \mathsf{Log}$ with $\forall N\,\forall n=\abs{N}$, and the sequences, instead of for example $\forall x_1,\dots,x_{p(n)}\in \{0,1\}^n$, having $$\forall X\: \left(\mathsf{Seq}(X)\land \mathsf{Len}(X)=p(n)\land \mathsf{Size}(X)=n\right),$$
where the functions $\mathsf{Seq}$, $\mathsf{Len}$ and $\mathsf{Size}$ are the same as in \Cref{sec:witnessing}.

\begin{proposition}\label{prop:phi}
    For $\varphi\in \Sigma^b_{i+1}$  and $\Phi_{p}$ defined as above with $p(n)=\bd(n)$ where $\bd$ is a symbol of the language, $\mathsf{T}^i_2+\mathsf{LIND}(\mathsf{s}\Sigma^b_{i+1},\bd)\vdash \Phi_{\bd}$.
\end{proposition}

\begin{proof}
Instead of using the length induction axiom, we use the equivalent length maximisation axiom on $\Sigma^b_{i+1}$ formulas, $\mathsf{LMAX}(\mathsf{s}\Sigma^b_{i+1},\bd)$, as seen in \Cref{prop:lind}. We define the following $\Sigma^b_{i+1}$ formula with parameters $x_1,\dots,x_{\bd(n)}$ (we can rewrite it as strict):
$$
\psi(t):=\exists y,i \;(\mathsf{Seq}(y) \land \mathsf{Seq}(i) \land \mathsf{Len}(y)=\mathsf{Len}(i)=t\land\forall j\leq t\; \varphi(x_{i_j},y_j)).
$$
This is that there exist a sequence of correct witnesses $y_j$ of length $t$ for the instances indexed by a sequence of indices $i$. The sentence $\psi(0)$ is trivially true, and since $n$ is a length, by $\mathsf{LMAX}(\mathsf{s}\Sigma^b_{i+1},\bd)$, we get:
\begin{equation}\label{eq:maximisation}
   \exists t \leq \bd(n) \;\big(\psi(t)\;\land\; \forall u\leq \bd(n), \; u>t\rightarrow \neg\psi(u)\big)
\end{equation}

We can argue in $\PV_1$, that there is a function $\mathsf{Compose}(y,i,l)$, where $y$ and $i$ are sequences of equal length and the elements of $i$ are between $1$ and the natural number $l$. This function will construct a sequence $Y$ of length $l$, where $Y_{i_j}=y_j$ for all elements of the sequence $i$, and for $j\in [l]\setminus i$, $Y_j=0$. From the first part of \Cref{eq:maximisation}, there exist $y_t$ and $i_t$, such that if $Y=\mathsf{Compose}(y,i,\bd(n))$,
$$
\forall j\in i_t\;\varphi(x_j,Y_j).
$$
However, the second part gives us that if we extend $y_t$ and $i_t$, the sentence will not be true any more. We use the function $\mathsf{Ext}(w,x)$ which extends the sequence $w$ by putting the element $x$ in the end. Then, by $u>t\rightarrow\neg \psi(u)$, we get
$$
\forall j\not\in i_t,\forall z,\; \mathsf{Len}(\mathsf{Ext}(y_t,z))=\mathsf{Len}(\mathsf{Ext}(i_t,j))>t\implies $$
$$\forall j\not\in i_t,\forall z,\; \exists k\in i_t\cup\{j\}\;  \neg \varphi(x_{\mathsf{Ext}(i_t,j)_k},\mathsf{Ext}(y_t,z)_k) \implies $$
$$\forall j\not\in i_t,\forall z,\; \neg\varphi(x_j,z).
$$
Combining the two results about $j\in i_t$ and $j\not \in i_t$, we conclude that 
$$
\forall j,\forall z\leq \bd(n) \; \neg\varphi(x_j,z) \lor \varphi(x_j,Y_j),
$$
which implies $\Phi_\bd$.
\end{proof}

\begin{proposition}\label{prop:psi}
    For $\varphi\in \Sigma^b_{i+1}$  and $\Psi_{p}$ defined as above with $p(n)=\bd(n)$ where $\bd$ is a symbol of the language, $\mathsf{T}^i_2+\mathsf{BB}(\mathsf{s}\Sigma^b_{i+1},\bd)\vdash \Psi_{\bd}$.
\end{proposition}

\begin{proof}
    This is an immediate application of the bounded replacement axiom. For $X=(x_1,\dots,x_{\bd(n)})$, we define $\psi(X,i,y):=\varphi(x_i,y)$, where $\psi$ is also $\Sigma^b_{i+1}$. Then by the axiom on $\psi$, we get:
    $$
    \mathsf{BB}(\psi,\bd) = \forall n\in \mathsf{Log}\; \forall X\; (\forall i\leq \bd(n)\;\exists y\;\psi(X,i,y))\rightarrow \exists Y\;\forall j\leq \bd(n) \psi(X,j,Y_j),
    $$
    but this is equivalent with $\Psi_{\bd}$.
\end{proof}

We will now go to the proofs \Cref{thm:bbsep,thm:lindsep,thm:lindbbsep}. Assume we have two increasing functions $\bd_1(x), \bd_2(x) \leq x$, such that for any $k\in\N$ and for large enough $x$, $\bd_1(x)> \bd_2(x)^k$. The input size of $\Phi_{\bd_1}$ and $\Psi_{\bd_1}$ is $n\cdot \bd_1(n)$. But we have
\begin{equation}\label{eq:sizeinequality}
\bd_2(n\cdot\bd_1(n))^k<\bd_1(n\cdot\bd_1(n))^{1/k_0}\leq \bd_1(n^2)^{1/k_0} < \bd_1(n)
\end{equation}
The first inequality is from the assumptions, the second from the fact that $\bd_1$ is increasing and sublinear. For the third inequality, we claim that there is some $k_0$ such that it holds, hence we get the desired result. It is easy to check the validity of the inequality $\bd_1(n^2)< \bd_1(n)^{k_0}$ for large enough $n$, if $\bd_1$ is one of the usual sublinear functions, e.g. $n^{1-\varepsilon},\,\log n,$ etc. For a full proof, see \Cref{sec:growth}.

After we have fixed the above, we can continue to the three proofs.

\begin{proof} [Proof of \Cref{thm:bbsep}] Assume that $\mathsf{T}^i_2+\mathsf{BB}(\mathsf{s}\Sigma^b_{i+1},\bd_1)\subseteq \mathsf{T}^i_2+\mathsf{BB}(\mathsf{s}\Sigma^b_{i+1},\bd_2)$. From \Cref{prop:psi}, we know that $\mathsf{T}^i_2+\mathsf{BB}(\mathsf{s}\Sigma^b_{i+1},\bd_1)\vdash \Psi_{\bd_1}$, which implies that $\mathsf{T}^i_2+\mathsf{BB}(\mathsf{s}\Sigma^b_{i+1},\bd_2)\vdash \Psi_{\bd_1}$, as well. This means by the witnessing theorem (\Cref{thm:bbwitness}) that $\Psi_{\bd_1}\in \mathcal{ST}^{\Sigma^b_i}[O(1), \poly (\bd_2(m))]$, where $m=n\cdot \bd_1(n)$ is the input size. However, from (\ref{eq:sizeinequality}), we have that $\poly(\bd_2(m))< \bd_1(n)$, and, from \Cref{thm:st2}, $\Psi_{\bd_1}\not\in \mathcal{ST}^{\Sigma^b_i}[r(m)+1,q(m)]$ if $r(m)q(m)<\bd_1(n)$, which shows that $\Psi_{\bd_1}\not\in \mathcal{ST}^{\Sigma^b_i}[O(1),\poly(\bd_2(m))]$, thus we get a contradiction.
\end{proof}

\begin{proof}[Proof of \Cref{thm:lindsep}] We know from \Cref{prop:phi} that $\mathsf{T}^i_2+\mathsf{LIND}(\mathsf{s}\Sigma^b_{i+1},\bd_1)\vdash \Phi_{\bd_1}$. On the other hand, by the proof of \Cref{thm:st1}, if $r(m)<\bd_1(n)$ for $m=n\cdot\bd_1(n)$, then $\Phi_{\bd_1}\not\in \mathcal{ST}^{\Sigma^b_i}[r(m),\poly(m)]$, which gives us from (\ref{eq:sizeinequality}) that $\Phi_{\bd_1}\not\in \mathcal{ST}^{\Sigma^b_i}[\poly(\bd_2(m)),1]$. This comes in contradiction with $\mathsf{T}^i_2+\mathsf{LIND}(\mathsf{s}\Sigma^b_{i+1}, \bd_2)\vdash \Phi_{\bd_1}$ by the witnessing theorem (\Cref{thm:lindwitness}).
\end{proof}

\begin{proof}[Proof of \Cref{thm:bblindsep}] From \Cref{prop:phi}, we have that $\mathsf{T}^i_2+\mathsf{LIND}(\mathsf{s}\Sigma^b_{i+1},\bd)\vdash \Phi_{\bd}$, while from \Cref{thm:st1}, $\Phi_{\bd}\not\in \mathcal{ST}^{\Sigma^b_i}[O(1),\poly(m)]$. Thus, if we assume that $\mathsf{T}^i_2+\mathsf{BB}(\mathsf{s}\Sigma^b_{i+1},\poly)\vdash \Phi_{\bd}$, we get a contradiction by \Cref{thm:bbwitness}.
\end{proof}

\begin{proof}[Proof of \Cref{thm:lindbbsep}] From \Cref{prop:psi}, we have that  $\mathsf{T}^i_2+\mathsf{BB}(\mathsf{s}\Sigma^b_{i+1},\bd_1)\vdash \Psi_{\bd_1}$, while from \Cref{thm:st2} and from (\ref{eq:sizeinequality}), $\Psi_{\bd_1}\not\in \mathcal{ST}^{\Sigma^b_i}[\poly(\bd_2(m)),1]$. Thus, if we assume that $\mathsf{T}^i_2+\mathsf{LIND}(\mathsf{s}\Sigma^b_{i+1},\bd_2)\vdash \Psi_{\bd_1}$, we get a contradiction by \Cref{thm:lindwitness}.
\end{proof}

%% file: unprovability.tex
\section{Unprovability Results}\label{sec:unprovability}

\subsection{Unprovability of circuit upper bounds}

In this section we will extend the result of Krajíček and Oliveira~\cite{KO17} to the theory $\PV_1+\mathsf{BB}(\Sigma^b_1)$, we will proceed as they do, but with some tweaks to adapt the proof for the stronger theory.\footnote{Let us note, that a more direct proof of this unprovability was developed in the follow-up work of Carmosino, Kabanets, Kolokolova and Oliveira~\cite{CKKO}. The authors were able to adapt the proof from~\cite{KO17}, which involves a greater degree of manipulation of the theories in question. Considering the more direct proof in the context of our witnessing theorems could be of further interest.}

\begin{definition}
Let $c,k\geq 1$ and let $p$ be a unary $\PV$-symbol. We say a $\PV$-symbol $f$ (possibly with some fixed parameters) describes a uniform sequence of $\abs{p(1^{(n)})}$-size families of circuits of size $c\cdot n^k$ if on the input $1^{(n)}$, it produces a sequence of $5$-tuples 
\[(i,1^{(n)},u,v,w),\]
where $i$ is the index of the circuit bounded by $\abs{p(1^{(n)})}$, $u,v$ are the indices of the nodes all of which have length polynomial in $n$, $w$ is the description for the gate type of $v$ (or the description that it is an input node). We assume that the output node is determined by a special tuple. The number of the $5$-tuples which start with any given $i$ is bounded by $c\cdot n^k$. And for each $i\leq \abs{p(1^{(n)})}$ we will use $f_i(1^{(n)})$ to denote the sequence of all tuples starting with $i$.
\end{definition}

The following is straightforward, so we omit the proof.

\begin{lemma}\label{lemmasuccrep}
    Let $c,k\geq 1$, let $p$ be a unary $\PV$-symbol and let $f$ be a $\PV$-symbol that describes a uniform sequence of $\abs{p(1^{(n)})}$-size families of circuits of size $c\cdot n^k$. Then there is a $\PV$-symbol $\tilde f$ that decides the language $L_{succ}'$ defined below.

    Assume that $L_{dc}'$ is the direct connectivity language of the sequence of $\abs{p(1^{(n)})}$-sized families of circuits described by $f$, that is the language of all the $5$-tuples in the description of the family for any input size. We let $L_{succ}'$ be the language that for each $(i,1^{(n)},u,v,w)\in L_{dc}'$ contains the $6$-tuple
    \[(i,Bin(n)01^{(n^{1/(3k)})},u,v,w,1^{(t)}),\]
    where $Bin(n)$ is the dyadic numeral for $n$ which is at most constantly longer than $\abs{n}$ and $t$ is chosen to pad the length of the $6$-tuple to length exactly $\ceil{n^{1/(2k)}}$.
\end{lemma}

\begin{definition}\label{definitionup}
    Let $c,k\geq 1$ and let $f$ be a unary $\PV$-symbol. We define the sentence $\UP_{k,c}(f)$ as
    \[(\forall 1^{(n)})(\exists C_n,\text{ a circuit of size at most $c\cdot n^k$})(\forall x,\abs{x}=n)(f(x)=1\leftrightarrow \CircuitVal(C_n,x)=1),\]
    where $\CircuitVal(C,x)$ is a $\PV$-symbol which evaluates the circuit encoded by $C$ on the input $x$.
\end{definition}

The following lemma was already proved by Krajíček and Oliveira in the strength which is sufficient for our purpose, so we omit its proof.

\begin{lemma}[{{\cite[Lemma~3.1]{KO17}}}]\label{lemmatimehierarchy}
    For every $d\geq 1$ there is a $\PV$-symbol $g_{d+1}$ deciding a language in $\DTIME[n^{d+1}]$, such that for every $\PV$-symbol $h$ deciding a language in $\DTIME[n^d]$ with advice of length $n^{2/3}$ there is $c_h\geq 1$, such that 
    \[\PV_1\vdash (\forall 1^{(n)},n\geq c_h)(\forall a,\abs{a}=n^{2/3})(\exists x,\abs{x}=n)(h(x,a)\neq g_{d+1}(x)).\]
\end{lemma}

\begin{lemma}
    Let $c,k\geq 1$, let $p$ be a unary $\PV$-symbol and let $f$ be a $\PV$-symbol that describes a uniform sequence of $\abs{p(1^{(n)})}$-size families of circuits of size $c\cdot n^k$ and let $\tilde f$ be the $\PV$-symbol from the statement of Lemma~\ref{lemmasuccrep}. Assume that \[\PV_1+\mathsf{BB}(\Sigma^b_1)\vdash \UP_{k,c}(\tilde f)\] and let $g$ denote $g_{3k+3}$ from Lemma~\ref{lemmatimehierarchy}. Then there exists $c_f$ such that 
    \[\PV_1+\mathsf{BB}(\Sigma^b_1)\vdash (\forall 1^{(n)},n\geq c_f)(\forall i\leq \abs{p(1^{(n)})})(\exists x,\abs{x}=n)(g(x)\neq \CircuitVal(f_i(1^{(n)}),x)).\]
\end{lemma}
\begin{proof}
    Assume that in a model of $\PV_1+\mathsf{BB}(\Sigma^b_1)$ the provability fails, that is 
    \[(\forall c_f)(\exists 1^{(n)},n\geq c_f)(\exists i\leq \abs{p(1^{(n)})})(\forall x,\abs{x}=n)(g(x) =\CircuitVal(f_i(1^{(n)}),x))\tag{*}\]
    we will argue that in such a model there is a polynomial-time algorithm $h$ which will certify that $g$ is infinitely often in $\DTIME[n^{3k+2}]/n^{2/3}$.

    Let $c_f$ be arbitrary, and let $i$ and $n$ be such that (*) is valid. By our assumption $\UP_{k,c}(\tilde f)$, we know that there are circuits recognizing the language of the $6$-tuples
    \[(i,Bin(n)01,u,v,w,1^{(t)}),\]
    decided by circuits $D_m$ on the input size $m=\ceil{n^{1/(2k)}}$. On the input $x$ of length $n$ the algorithm $h$ tries all possible $3$-tuples $(u,v,w)$ and for each of them runs the $D_m$ on the corresponding $6$-tuple where $i$ is taken from~$(*)$, this is at most $O(n^{2k+1})$ possibilities and it takes time $O(n^{2k+2})$ to run $D_m$ on all of them. Thus by $\UP_{k,c}(\tilde f)$, the algorithm $h$ obtains a description of a circuit of size at most $c\cdot n^k$ and uses it to compute a value for $g(x)$ in time $O(n^{2k})$. Thus, $g$ can be computed by $h\in \DTIME[n^{2k+2}]/n^{2/3}$ where the advice contains the description of $D_m$. This is in contradiction with Lemma~\ref{lemmatimehierarchy}.
\end{proof}

\begin{lemma}\label{lemmarefuter}
    Let $f,\tilde f, g$ and $p$ be as in the previous lemma. Then there exists $c_f$ and a $\PV$-symbol $e$ such that 
    \[\PV_1+\mathsf{BB}(\Sigma^b_1)\vdash (\forall 1^{(n)},n\geq c_f)(\forall i\leq \abs{p(1^{(n)})})(g(e(i,1^{(n)}))\neq \CircuitVal(f_i(1^{(n)}),e(i,1^{(n)}))).\]
\end{lemma}
\begin{proof}
    Follows straightforwardly from Herbrand's theorem.
\end{proof}

\begin{theorem}\label{theoremkobb}
    For every $k\geq 1$ there is a unary $\PV$-symbol $h$ such that for every constant $c\geq 1$
    \[\PV_1+\mathsf{BB}(\Sigma^b_1)\nvdash \UP_{k,c}(h).\]
\end{theorem}
\begin{proof}
    We will proceed by contradiction. Assume that there is a $k$ where for any unary $h$ there is a constant $c$ such that \[\PV_1+\mathsf{BB}(\Sigma^b_1)\vdash \UP_{k,c}(h).\]

    In particular, there is a constant $c\geq1$ such that $\PV_1+\mathsf{BB}(\Sigma^b_1)\vdash \UP_{k,c}(g).$ By Theorem~\ref{thm:bbwitness}, there is a polynomial-time student $s$ which generates sets of polynomially many answers for the candidate circuit computing the symbol $g$ and after $r$-many rounds of correction from the teacher it always outputs at least one circuit which computes $g$ on the input length $1^{(n)}$. 
    
    Assume that $s^1,\dots,s^{r}$ are the $\PV$-symbols computing the set of answers of the student from the teachers answer for each of the $r$-many rounds, these can be straightforwardly constructed from the subterms of $s$. That is, there is a $\PV$-symbol $p$ such that
    \begin{align*}
        \mathbb{N}\models \quad&(\exists i\leq \abs{p(1^{(n)})})(g((x_1)_i)=\CircuitVal(s^1_i(1^{(n)}),(x_1)_i))\\
                          \lor&(\exists i\leq \abs{p(1^{(n)})})(g((x_2)_i)=\CircuitVal(s^2_i(1^{(n)},x_1),(x_2)_i))\\
                          &\quad\vdots\\
                          \lor&(\exists i\leq \abs{p(1^{(n)})})(g((x_r)_i)=\CircuitVal(s^r_i(1^{(n)},x_1,\dots,x_{r-1}),(x_{r})_i)),
    \end{align*}
    where $(-)_i$ denotes the $\PV$-symbol which extracts the $i$-th element of a sequence. In the rest of the proof we will show that from the assumption, that there is a $\tilde c_1$ such that \[\PV_1+\mathsf{BB}(\Sigma^b_1)\vdash \UP_{k,\tilde c_1}(s^1(1^{(n)})),\] we can invoke Lemma~\ref{lemmarefuter} to falsify the first disjunct. Repeating this $r$-many times, we obtain a contradiction.
    
    By our assumption, there is a $\tilde c_1\geq 1$ such that $\PV_1+\mathsf{BB}(\Sigma^b_1)\vdash \UP_{k,\tilde c_1}(s^1(1^{(n)}))$, by Lemma~\ref{lemmarefuter} there is a $\PV$-symbol $e^1$ and $c_{1}$, such that
    \[\PV_1+\mathsf{BB}(\Sigma^b_1)\vdash (\forall 1^{(n)},n\geq c_1)(\forall i\leq \abs{p(1^{(n)})})(g(e^1(i,1^{(n)}))\neq\CircuitVal(s^1_i(1^{(n)}),e^1(i,1^{(n)}))),\]
    by substituting $\tilde e^1(1^{(n)})$ for $x_1$, where $\tilde e^1(1^{(n)})$ is the $\PV$-symbol computing the sequence \[(e^1(i,1^{(n)}))_{i\leq \abs{p(1^{(n)})}},\] we falsify the first disjunct and obtain
    \begin{align*}
        \mathbb{N}\models \quad&(\exists i\leq \abs{p(1^{(n)})})(g((x_2)_i)=\CircuitVal(s^2_i(1^{(n)},\tilde e^1(1^{(n)})),(x_2)_i))\\
                          \lor&(\exists i\leq \abs{p(1^{(n)})})(g((x_3)_i)=\CircuitVal(s^3_i(1^{(n)},\tilde e^1(1^{(n)}),x_2),(x_3)_i))\\
                          &\quad\vdots\\
                          \lor&(\exists i\leq \abs{p(1^{(n)})})(g((x_r)_i)=\CircuitVal(s^r_i(1^{(n)},\tilde e^1(1^{(n)}),x_2,\dots,x_{r-1}),(x_{r})_i)).
    \end{align*}

    To proceed further, we again use our assumption to obtain $\tilde c_2\geq 1$ such that 
    \[\PV_1+\mathsf{BB}(\Sigma^b_1)\vdash \UP_{k,\tilde c_2}(s^2(1^{(n)},\tilde e^1(1^{(n)}))),\]
    and apply Lemma~\ref{lemmarefuter} again. After $r$-many invokation of Lemma~\ref{lemmarefuter}, we obtain that the empty disjunction is true in $\mathbb{N}$ which is a contradiction.
\end{proof}

\subsection{Unprovability of average-case circuit lower bounds}

The following result closely follows the result in \cite{pichS21}, where we do all the modifications needed for the corresponding theory and the Student-Teacher Game. Another very helpful exposition for the proof can be found in \cite{lioliveira}.

We express the average-case lower bound $\mathsf{NSubExp} \not\subseteq \mathsf{Avg-} \mathsf{coNSIZE}[2^{n^\delta}]$ in an instance manner. That is we assume we have a nondeterministic Turing machine $M$ whose running time is a constructive function $t(n)$ with $t(n)=2^{o(n)}$. The lower bound states that there is no co-nondeterministic circuit $D$ of size $s(n)$, where we will have $s(n)=2^{n^\delta}$, which approximates efficiently $M$, in the sense that it outputs the same for at least $1/2+2^{-n^\delta}$ fraction of the inputs. We can express that with the following sentence in the language of $\PV$:
\begin{align*}
\mathsf{LB}(M,s,m,n_0):= & \;\forall n\geq n_0\in \mathsf{LogLog}\;\; \forall D\in \mathsf{coNSIZE}[s(n)]\\
& \exists m=m(n), (x_1,\dots,x_{m})\in (\{0,1\}^n)^m\; \forall i\leq m(n)\;\mathsf{Error}_{M,D}(x_i)
\end{align*}
where we define the formula
$$
\mathsf{Error}_{M,D}(x) = \big[\exists y,z \; M(x,y)=1\land D(x,z)=0 \big]\; \lor\;  \big[\forall y',z' \; M(x,y)=0\land D(x,z)=1 \big],
$$
which means that $M(x)\neq D(x)$. Note, that the size of $y,z,y',z'$ are at most $2^{n^\delta}$.

\begin{theorem}\label{theorempsllind}
    For every $n_0\in \N$, $\delta\in \mathbb{Q} \,\cap\, (0,1)$, and a non-deterministic Turing machine $M$ with running time $t(n)=2^{o(n)}$, it holds that
    $$
    \PV_1+\mathsf{LLIND}(\mathsf{s}\Sigma^b_1)\not\vdash \mathsf{LB}(M,2^{n^\delta},2^n/2-2^n/2^{n^\delta},n_0).
    $$
\end{theorem}

\begin{proof}
    Assume that $\PV_1+\mathsf{LLIND}(\mathsf{s}\Sigma^b_1)\vdash \mathsf{LB}(M,2^{n^\delta},2^n/2-2^n/2^{n^\delta},n_0)$. We have two consequences:
    \begin{enumerate}
        \item $\PV_1+\mathsf{LLIND}(\mathsf{s}\Sigma^b_1)\vdash \mathsf{LB}(M,2^{n^\delta},1,n_0)$, which means that $\PV_1$ also shows the worst case lower bound; for all co-nondeterministic circuits $D$, there is at least one input, where $M$ and $D$ do not agree.
        \item $\N\models \mathsf{LB}(M,2^{n^\delta},2^n/2-2^n/2^{n^\delta},n_0)$, which means that in the standard model there is no co-nondeterministic circuit $D$ that can approximate $M$ with success rate greater than $1/2+2^{-n^\delta}$.
    \end{enumerate}
    In this proof, we assume the first point to find a contradiction with the second point.

    From the witnessing theorem for the theory $\PV_1+\mathsf{LLIND}(\mathsf{s}\Sigma^b_1)$ and the first point, we get that the formula $\mathsf{LB}(M,2^{n^\delta},1,n_0)$ can be witnessed by a Student-Teacher Game with poly-logarithmic number of rounds. However, since $n\in \mathsf{LogLog}$, this means that the length of the input is $2^{O(n)}$, which gives us $n^k$ rounds for some $k\in \N$. This $k$ is fixed from now on, since it depends only from the proof in $\PV_1+\mathsf{LLIND}(\mathsf{s}\Sigma^b_1)$, and we will refer to the number of rounds as $r:=n^k$.

    We simulate the Student by the polynomial-time function $f$, which takes as input the number $2^n$ in unary, the description of a co-nondeterministic circuit $D$, the number of round $i$ and depending on $i$, $i-1$ counterexamples in the form of $y',z'$. On every round, $f$ runs in time $2^{O(n)}$.
    After each round, $f$ outputs a triplet $(x,y,z)$, where $x$ is the candidate instance for $M(x)\neq D(x)$, and $y,z$ are the candidate witnesses for the case that $M(x)=1$ and $D(x)=0$, respectively. We will denote the outputs for the respective round $i\leq r$ as $(x_i,y_i,z_i)$.

    We use as input for the Student-Teacher Game, a specific type of co-nondeterministic circuit, $D_{w}(x)$, which are defined using the Nisan-Wigderson generator \cite{nw}. We fix some $c\geq \max(2\delta^{-1}(k+1),4)$ and for some $m'$ with $n^c\leq m'\leq 2n^c$, we consider a boolean $2^n\times m'$ matrix $A=\{a_{i,j}\}$. For every $i\in [2^n]$, we define the set $J_i(A)$ of all $j\in [m']$ such that $a_{i,j}=1$ (the $1$'s in the $i$-th row). Then, $A$ is a $(n,n^{c/2})$-design if for all $i$, $|J_i(A)|=n^{c/2}$ and, for all $i\neq j$, $|J_i(A)\cap J_j(A)|\leq n$. In \cite{nw}, they construct such a design, and they show that there is $n^{9c/2}$-size circuit that given $w\in \{0,1\}^{m'}$ and $i\in [2^n]$, outputs the restriction of $w$ on the bits indexed by $J_i(A)$. Then, we define the circuits $D_w$ as
    $$
    D_w(x)=\neg M(w|J_x(A)).
    $$
    If we assume that $M$ is computable by a nondeterministic circuit of size $2^{n^\epsilon}$, then $D_w(x)$ can be computed in co-nondeterministic size $2^{n^{c\epsilon}}$, which is less than $2^{n^\delta}$, if we choose $\epsilon< \delta /2c$.

    From now on, we work with a specific size $n$, so there is no ambiguity with indexing the rounds with numbers which are a function of $n$. Particularly, $r=n^k$ is also fixed.

    We will use some more definitions. For $a\in\{0,1\}^{m'-n^{c/2}}$ and $b\in \{0,1\}^{n^{c/2}}$, we define $r_x(a,b)$ as the bit-string of length $m'$ with the bits of $b$ into the positions indexed by $J_x(A)$ and the bits of $a$ in the remaining positions. We also define the trace of a string $w\in \{0,1\}^{m'}$ as the sequence of the $x$-outputs of the Student during the Student-Teacher Game with input the circuit $D_w$. For this, we need to specify the outputs of the Teacher (e.g. the Teacher outputs the minimum counterexample lexicographically). Then, $\mathsf{tr}(w)=(x_1,x_2,\dots,x_t)$ with $t\leq r$, where for $x_t$ we have that the Student succeeds, in other words $M(x_t)\neq D(x_t)$.

    \begin{lemma}\label{lem:trace}
        There is a trace $T=(X_1,X_2,\dots,X_t)$ with $t\leq r$ and some $a\in\{0,1\}^{m'-n^{c/2}}$, such that considering all the traces of the form $\mathsf{tr}(r_{X_t}(a,y))$ for all $y\in \{0,1\}^{n^{c/2}}$, a fraction of $s\geq 1/(n3^{t-1}2^{nt})$ of them start with $T$ and a fraction of at least $(2/3-1/n)s$ of them are equal to $T$.
    \end{lemma}

    We do not revise the proof of this lemma, since there are already proofs in \cite{DBLP:journals/jml/Krajicek11,PICH201529,pichS21} and a detailed proof of the averaging argument in Appendix E of \cite{lioliveira}.

    We fix the trace $T$ and the string $a$ from the lemma, and with their help, we will define a new co-nondeterministic circuit, $\mathcal{D}$ that approximates $M$ for the input size $n^{c/2}$. We want this $\mathcal{D}$ to simulate the Student-Teacher Game for the input $n$, but in order to achieve this, we need to simulate the Teacher's answer, as well. This is where we use the property of the Nisan-Wigderson design.

    We only care about instances of the Game, where the input co-nondeterministic circuit is in the form $D_w$ with $w=r_{X_t}(a,b)$ for some $b\in \{0,1\}^{n^{c/2}}$. From the structure of the initial sentence, we see that for any $x\in\{0,1\}^n$, the Teacher must determine counterexamples $y_x',z_x'$ such that $M(x,y_x')=1$ or $D_w(x,z_x')=0$. If there is a an appropriate $y_x'$, then this has size at most $2^n$ and it is independent from the input of the Game, $D_w$. If there is no such $y_x'$, then the Teacher must find some $z_x'$, such that $D_w(x,z_x')=0$, or equivalently, $M(w|J_x(A),z_x')=1$. However, there are only $n^{c/2}$ bits of $w$ which are not fixed by $a$, and from those only $|J_x(A)\cap J_{X_t}(A)|\leq n$ actually contribute for the function, which we get by the $(n,n^{c/2})$-design. This means we can describe the function $x\mapsto z_x'$ with $2^{O(n)}$ bits, too. In total, the answer of the Teacher for each $x$ can be described by some advice $Y_x$ of size $2^{O(n)}$. Thus, the total advice we need to simulate the Teacher for polynomial number of different inputs is still of size $2^{O(n)}$.

    Next, we describe a co-nondeterministic algorithm with advice of size $2^{O(n)}$, which is equivalent with a co-nondeterministic circuit of size $2^{O(n)}$, denoted $\mathcal{D}$. The input must have size $n^{c/2}$.

\begin{algorithm}[H]
     \SetKwInOut{Input}{Input}
     \SetKwInOut{Advice}{Advice}
     \caption{The pseudocode of the algorithm $\mathcal{D}$ that approximates $M$ on input sizes $n^{c/2}$.}
     \Input{A string $y\in \{0,1\}^{n^{c/2}}$.}
     \Advice{The trace $T=(X_1,\dots,X_t)$, $a\in \{0,1\}^{m'-n^{c/2}}$, $b_0,b_1\in \{0,1\}$ and $Y_x$ for all $x\in T$.}
     \vspace{0.3em}
     Define $w=r_{X_t}(a,y)$\;
     Simulate the Student-Teacher Game on input $D_w(x):=\neg M(w|J_x(A))$ until the round $t\leq r$ using the Teacher's answers $Y_x$.\;
     If at any point of the simulation, the trace of $w$ does not match $T$, then output $b_0$\;
     \tcp{We take $b_0$ to be the majority bit among the values of $M$ on all the inputs that do not have $T$ as prefix to their trace.}
     If the simulation works correctly until round $t$, output $b_1$.
     \tcp{We take $b_1= M(X_t)$.}
\end{algorithm}

We will show that $\mathcal{D}$ correctly approximates $M$ on input sizes $n^{c/2}$. For input $y$, if the trace of $r_{X_t}(a,y)$ does not have $T$ as a prefix, we know that the algorithm outputs the majority bit $b_0$. In this case, the algorithm is correct with probability at least $1/2$. On the other hand, we know from \Cref{lem:trace} and the selection of the advice, that for at least $1/(n3^{t-1}2^{nt})$ fraction of the inputs, the corresponding trace has $T$ as prefix, and for $(2/3-1/n)$ fraction of them, their trace is exactly $T$. 

When the latter is the case, the algorithm correctly computes $M(y)$, since from the definition of the trace, the last element is when the Student succeeds in the Game, which means that $M(X_t)\neq D_w(X_t)$. From the advice beat $b_1= M(X_t)$, we get that 
$$
\mathcal{D}(y)=b_1=\neg D_w(X_t)=M(w|J_{X_t}(A))\overset{w=r_{X_t}(a,y)}{=\joinrel=\joinrel=\joinrel=\joinrel=} M(y).
$$

Therefore, in the case of inputs that have $T$ as prefix in their trace, the probability that $\mathcal{D}$ is correct is at least $2/3-1/n$. Overall,
$$
\Pr[\mathcal{D}(y)=M(y)]\geq \frac{1}{2}+\left(\frac{2}{3}-\frac{1}{n}\right)\frac{1}{n3^{t-1}2^{nt}}\geq \frac{1}{2}+\Omega\left(\frac{1}{2^{n^{k+1}}}\right).
$$
The second inequalty is due to the fact that $t\leq r=n^k$.

Since the input size is $n^{c/2}$, we get a contradiction if $2^{n^{k+1}}\leq 2^{(n^{c/2})^\delta}$, which means that we must have $c\geq 2\delta^{-1}(k+1)$, which is the case for our choice of $c$.
\end{proof}

\subsection{Unprovability of both in a single theory}

Both of the previous section used the structure of the Student-Teacher protocol related to the theory. On one hand, it seems to be challenging to adapt unprovability of Krajíček and Oliveira for logarithmically many adaptive rounds of the student. On the other hand it seems downright impossible to use any known witnessing theorems to extend the unprovability reuslt of Pich and Santhanam to polynomially many queries per round, as these queries can cover all possible values for the witness. It is not hard to see that the theory $\PV_1+\mathsf{BB}(\mathsf{s}\Sigma^b_1,\log)$ is contained in both $\PV_1+\mathsf{BB}(\Sigma^b_1)$ and $\PV_1+\mathsf{LLIND}(\mathsf{s}\Sigma^b_1)$ and thus both unprovability results hold for this theory, which is by our results stronger than $\PV_1$ assuming $\NP\not \subseteq \P/poly$.

\begin{corollary}\label{corollarybothunprovabilities}
    \begin{enumerate}
        \item[]
        \item For every $k\geq 1$ there is a unary $\PV$-symbol $h$ such that for every constant $c\geq 1$
            \[\PV_1+\mathsf{BB}(\mathsf{s}\Sigma^b_1,\log)\nvdash \UP_{k,c}(h).\]
        \item For every $n_0\in \N$, $\delta\in \mathbb{Q} \,\cap\, (0,1)$, and a non-deterministic Turing machine $M$ with running time $t(n)=2^{o(n)}$, it holds that
        $
        \PV_1+\mathsf{BB}(\mathsf{s}\Sigma^b_1,\log)\not\vdash \mathsf{LB}(M,2^{n^\delta},2^n/2-2^n/2^{n^\delta},n_0).
        $
        \item Assuming $\NP\not\subseteq \P/poly$, we have $\PV_1\not\equiv \PV_1+\mathsf{BB}(\mathsf{s}\Sigma^b_1,\log)$.
    \end{enumerate}
    \begin{proof}
        Part 1 follows immediately from Theorem~\ref{theoremkobb}. By Theorem~\ref{theoremllindprovesbb} we obtain that \[\PV_1+\mathsf{LLIND}(\mathsf{s}\Sigma^b_1) \vdash\mathsf{BB}(\mathsf{s}\Sigma^b_1,\log),\]
        and thus part 2 follows from Theorem~\ref{theorempsllind}. Part 3 is a direct application of Theorem~\ref{thm:bbsep}.
    \end{proof}
\end{corollary}

%% file: appendix.tex
\section{Function Calculations}\label{sec:appendix}

\subsection{Transformation of Input Sizes}\label{sec:input}

Here, we describe how to compute the function $p(n)$ given the input $n$ and show that it is actually a solution to the relations $m=n\cdot p(n)$ and $p(n)=g(m)$, for a function $g: \N\to\N$, which has sublinear growth rate, it is unbounded and increasing. 

For the general statements of \Cref{thm:st1,thm:st2}, where no logical theory is involved, it is enough for us to show that there is an algorithm that determines the value $p(n)$ for every input $n$. In the theories we consider, we will specify $p(n)$ by some function in the language. This does not need to satisfy $p(n)=g(m)$, since for the separation of theories we do not need the fine-grained separation we prove in these general theorems.

We fix the input $n$, which we assume it is big enough. We need to find a natural number $k$, such that $k = g(n \cdot k)$. Then, $p(n)=k$. For that, we define the function
$$
h_n(k)=g(n\cdot k)-k,
$$
and now our goal is to show that for big enough $n$, the equation $h_n(k)=0$ always has a solution.

For $k=1$, we have that $h_n(1)=g(n)-1$, so it is unbounded and increasing, like $g(n)$. This means that for some $n_0$, for all $n\geq n_0$, $h_n(1)>0$.

On the other hand, we want to show that for all these $n$, $h_n(k)$ will be negative at some point. From the definition of $h_n(k)$, we have that
$$
\frac{h_n(k)}{k} = \frac{g(n\cdot k)}{k}-1 = n\cdot \frac{g(n\cdot k)}{n\cdot k}-1.
$$
However, if $k$ tends to infinity, then $\frac{g(n\cdot k)}{n\cdot k}$ tends to $0$, since $g$ is sublinear. This means that there is some big enough $k$, such that $h_n(k)$ eventually becomes negative.

We have showed that $h_n(1)>0$ and for some big value of $k$, $h_n(k)<0$. The only thing remaining is that the value $0$ is actually achieved. For that, we will show that when $h_n(k)>0$, then $h_n(k+1)\geq 0$, which means that before the change to negative values, we actually get the solution to $h_n(k)=0$.

Since $h_n(k)>0$ and $h_n(k)$ is natural, we get $h_n(k)\geq 1\implies g(n\cdot k)\geq k+1$. This means that
$$
h_n(k+1)= g(n\cdot (k+1))-(k+1) \geq g(n\cdot k)-(k+1)\geq 0,
$$
as we wanted. The first inequality comes from the fact that $g$ is increasing.

\subsection{Growth of Sublinear Functions}\label{sec:growth}

Suppose that we have two increasing functions $\bd_1(x), \bd_2(x) \leq x$, such that for any $k\in\N$ and for large enough $x$, $\bd_1(x)> \bd_2(x)^k$. Then, we want to show the inequality, for large enough $n_0$:
\begin{equation}
\forall n\geq n_0 \;\forall k \;\bd_2(n\cdot\bd_1(n))^k< \bd_1(n)
\nonumber
\end{equation}

From the assumption we have that 
$$
\forall n\geq n_0 \;\forall k \;\bd_2(n\cdot\bd_1(n))^k< \bd_1(n\cdot\bd_1(n)).
$$
We can change the above to show that also
$$
\forall n\geq n_0 \;\forall k \;\forall k_0\;\bd_2(n\cdot\bd_1(n))^k< \bd_1(n\cdot\bd_1(n))^{1/k_0}\leq \bd_1(n^2)^{1/k_0},
$$
where for the second inequality we use the fact that $\bd_1(n)\leq n$.

For the last step, we combine the above relation with
$$
\forall n\geq n_0  \;\exists k_0\; \bd_1(n^2)^{1/k_0}< \bd_1(n),
$$
which is true, because otherwise, for $\bd(n_0)>1$, $\exists n\geq n_0\;\forall k_0\; \bd_1(n^2)>\bd_1(n)^{k_0}$, which is a contradiction, since the right-hand side tends to infinity.